\documentclass[journal]{IEEEtran}
\usepackage{amsmath,amsfonts,amsthm}
\usepackage{array}
\usepackage[caption=false,font=normalsize,labelfont=sf,textfont=sf]{subfig}
\usepackage{textcomp}
\usepackage{stfloats}
\usepackage{url}
\usepackage{verbatim}
\usepackage{graphicx}
\usepackage{xcolor}
\usepackage{nomencl}
\usepackage{etoolbox}
\usepackage{ifthen}
\usepackage{booktabs}
\usepackage[ruled]{algorithm2e}

\usepackage{booktabs} 
\usepackage{amssymb}  

\usepackage{multirow}
\usepackage{bm}
\usepackage{amsmath}
\usepackage{makecell}
\usepackage{enumitem}

\allowdisplaybreaks 
\usepackage[subtle,tracking=normal]{savetrees}

\usepackage{cite}

\usepackage{xcolor}  

\makenomenclature
\def\BibTeX{{\rm B\kern-.05em{\sc i\kern-.025em b}\kern-.08em
    T\kern-.1667em\lower.7ex\hbox{E}\kern-.125emX}}
\usepackage{balance}

\newtheorem{theorem}{Theorem}

\newtheorem{prop}[theorem]{Proposition}
\theoremstyle{definition}
\newtheorem{remark}{Remark}

\begin{document}

\bstctlcite{IEEEexample:BSTcontrol}

\title{Process-Knowledge-Embedded Safe DRL  for \\ Real-Time Dispatch of Process Loads  in \\ Industrial Microgrids}

\author{Daniyaer Paizulamu, \IEEEmembership{Graduate Student Member, IEEE}, 
Lin Cheng, \IEEEmembership{Senior Member, IEEE}, \\Fashun Shi, Yuchi Zhang, Zhaoyang Dong, \IEEEmembership{Fellow, IEEE}
\thanks{
This work was supported by Science and Technology Program of Xinjiang Uyghur Autonomous Region (No. 2024B010012). (\textit{Corresponding author: Zhaoyang Dong.})}

\thanks{Daniyaer Paizulamu, Lin Cheng and Fashun Shi are with the State Key Laboratory of Power System Operation and Control, Department of Electrical Engineering, Tsinghua University, Beijing 100084, China (e-mail: dnyepzlm22@mails.tsinghua.edu.cn; chenglin@mail.tsinghua.edu.cn; shi\_fashun@tsinghua.edu.cn).}
\thanks{Yuchi Zhang and Zhaoyang Dong are with the Department of Electrical Engineering, City University of Hong Kong, Hong Kong (e-mail: yzhang7522-c@my.cityu.edu.hk; zydong@cityu.edu.hk).}}

\markboth{IEEE TRANSACTIONS ON SUSTAINABLE ENERGY,~Vol.~X, No.~X, XX Month~2026}
{Paizulamu \MakeLowercase{\textit{et al.}}: SPLs Real-Time Dispatch}

\IEEEaftertitletext{\vspace{-2\baselineskip}}

\maketitle

\begin{abstract}
Steelmaking process loads (SPLs) are flexible resources that enhance local renewable-energy utilization and reduce electricity procurement costs in industrial microgrids. However, strong multistage coupling makes current decisions affect subsequent feasibility, challenging conventional deep reinforcement learning to reduce costs while maintaining process feasibility throughout production. This paper proposes a
process-knowledge-embedded safe deep reinforcement learning framework for the real-time dispatch of SPLs in industrial microgrids. Specifically,
a lossless active-frontier action space is constructed, and a
process-distance-guided action-processing mechanism reallocates
excluded-action probabilities according to process distance and the
actor's safe-action preference. Recursive process feasibility is
established to guarantee admissible execution and feasible continuation.
Furthermore, the expected process-correction distance is incorporated
into PPO through a correction budget and a primal-dual update to
internalize process knowledge into the raw policy, while a derived bound
quantifies the raw policy's dependence on safety processing. Case studies using real-world data demonstrate zero process losses, electricity-cost reductions of 49.2\% and 25.9\%
relative to rule-based scheduling and rolling MILP, respectively, within an acceptable computation time.
\end{abstract}

\mbox{}

\vspace{-0.5cm}
\section{Introduction}\label{Introduction}

\IEEEPARstart{R}{eal-time} dispatch of steelmaking process loads
(SPLs) in response to fluctuations in local renewable generation and
grid electricity prices is essential for increasing on-site renewable
energy utilization and reducing electricity procurement costs in
industrial microgrids~\cite{TOP1, TOP4}. SPLs are typically characterized by
large power ratings, continuous operation, strong inter-stage coupling,
and high start-stop costs. Their flexibility potential coexists with
stringent production constraints~\cite{TOP2}. Therefore, safe and
efficient scheduling of SPLs is critical for converting process
flexibility into reliable demand-side regulation while preserving
production continuity and output requirements~\cite{TOP3}.

Existing studies on flexible process-load scheduling can be broadly classified into: Dispatching rule-based methods\cite{Rule}, process systems engineering models based on RTN\cite{RTN}, CRTN\cite{CRTN} formulations, rolling MILP/MPC-based methods\cite{MILP,MPC}, and reinforcement learning-based methods\cite{DRLtop1,DRLtop2}. The first two categories rely heavily on day-ahead forecasts and cannot readily embed process knowledge into online action generation, making them unsuitable for intra-day adjustment in response to electricity price variations or renewable generation fluctuations. Rolling optimization methods offer good interpretability and feasibility, but require repeated optimization in large-scale, highly constrained, and uncertain environments. The computational burden limits real-time applicability and makes it difficult to meet 5-min scale dispatch requirements.

Under conditions of severe uncertainty and high-frequency scheduling, Deep Reinforcement Learning (DRL) exhibits substantial potential for real-time decision-making~\cite{DRL9}. DRL has been applied to real-time pricing-based industrial facility management~\cite{DRL1}, discrete-manufacturing demand
response~\cite{DRL2}, and heavy-industry energy
management~\cite{DRL16}. In steel-related applications, DRL has been
used for multi-objective oxygen-system scheduling~\cite{DRL7},
hybrid discrete--continuous scheduling of oxygen systems
~\cite{DRL3}, and electric arc furnace (EAF)-based steel-plant scheduling under electricity
price and on-site renewable-generation uncertainties
~\cite{DRL8}. However, conventional DRL does not guarantee
satisfaction of hard process constraints.

To overcome the inherent limitations of traditional DRL in guaranteeing safety constraints, safe deep reinforcement learning (SDRL) provides a novel paradigm for complex industrial scheduling and energy management~\cite{SRL0}. In \cite{SRL1}, real-time pricing-driven manufacturing demand response was formulated as a constrained Markov decision process with a hybrid action space, wherein a cross-attention mechanism and a Lagrangian Soft Actor-Critic (SAC) framework were employed to handle hybrid action coupling and safety constraints. Relatedly, process knowledge was integrated into SDRL for the real-time scheduling of multi-product gas supply networks in \cite{SRL2}, while an evolution-assisted SDRL approach was utilized for real-time production optimization under uncertainties in industrial rotary kilns in \cite{SRL3}. Furthermore, SDRL was applied to the optimal control of metallurgical processes in \cite{SRL4}. In the domain of microgrid energy management \cite{SRL7}, SDRL was applied to secure and economic dispatch involving multi-energy coupling and network constraints, from the perspectives of spatio-temporal awareness \cite{SRL5} and physics-informed safety layers \cite{SRL6}, respectively.

However, existing SDRL methods provide only partial support for
the real-time dispatch of SPLs. Constraint-based approaches
mainly regulate expected cumulative violations, whereas action-level
mechanisms typically filter, project, or replace infeasible actions
without fully accounting for multistage process logic, process-correction
distances, or feasible process continuation. Process knowledge therefore
remains largely external to the raw policy, which may continue assigning
substantial probability to infeasible actions and require repeated safety
intervention. Penalty-based designs may further suppress violations
through under-production, fewer startups, or reduced switching, thereby
compromising production-quota fulfillment and yielding overly
conservative schedules.

This paper develops a process-knowledge-embedded safe deep
reinforcement learning (PK-SDRL) framework for the real-time dispatch
of SPLs under electricity-price and local renewable-generation
uncertainties, while maintaining multistage process feasibility and
meeting prescribed production requirements. The main contributions
are summarized as follows:

\begin{enumerate}

\item \textit{Process-Distance-Guided Action Processing:}
A process-distance-guided action-processing (PDG-AP) mechanism is
developed to reallocate excluded-action probabilities toward feasible
actions according to process distance and the actor's safe-action
preference. Recursive process feasibility is established to ensure
admissible execution and feasible process continuation.

\item \textit{Parameterized-Action PPO with Process-Knowledge
Internalization:}
The expected process-correction distance is incorporated into PPO
through a process-correction budget and a primal--dual update. The
derived raw-policy process-infeasibility bound quantifies and limits the
actor's dependence on safety processing.

\item \textit{PK-SDRL Dispatch Framework:}
A PK-SDRL framework is developed for real-time dispatch at a
5-min resolution under a two-part electricity tariff and local
renewable-generation uncertainties. The lossless active-frontier action space, parameterized discrete-continuous policy, state-transition shaping reward, and
terminal quota reward jointly coordinate economic operation, process
advancement, and daily quota fulfillment.

\item \textit{Ablation and Comparative Performance Evaluation:}
The framework is evaluated on a multiline EAF--LF--CC steel plant using
real-world renewable-generation data and two-part tariff. Ablation, comparative, and sensitivity studies assess process feasibility,
quota fulfillment, economic performance, online decision time, and
robustness to forecast errors.

\end{enumerate}

The remainder of this paper is organized as follows.
Section~\ref{Modeling} introduces the process-constrained modeling of
SPLs. The methodology of the proposed PK-SDRL framework is presented
in Section~\ref{method}. Section~\ref{Case Study} validates the proposed method using real-world data. Finally, conclusions are summarized in Section~\ref{CONCLUSION}.

\section{Process-Constrained Modeling of SPLs}\label{Modeling}

This section first characterizes the dispatch modes of industrial process loads and then formulates a process-constrained dispatch model for steelmaking process loads (SPLs) in a short-process steelmaking plant following the electric arc furnace (EAF)--ladle furnace (LF)--continuous casting (CC) route.

\subsection{Dispatch Modes for Process Loads}
\label{Dispatch Modes}

According to their controllability in the power and temporal dimensions, process loads (PLs) can be classified into three fundamental dispatch modes: adjustable, transferable, and shiftable, as illustrated in Fig.~\ref{response_mode}. Adjustable loads vary their operating power within admissible limits during processing. Transferable loads allow production tasks to be reassigned among eligible parallel units or production lines. Shiftable loads permit task start times to be advanced or delayed within process-feasible time windows. These modes may coexist within the same production stage and jointly determine its dispatch capability. For example, an EAF batch may combine adjustable operating power with a shiftable start time, while its admissible response remains constrained by downstream refining, transfer, and continuous-casting requirements.

\begin{figure}[htbp] 
    \vspace{-0.2cm}
    \centerline{\includegraphics[width=0.9\columnwidth]{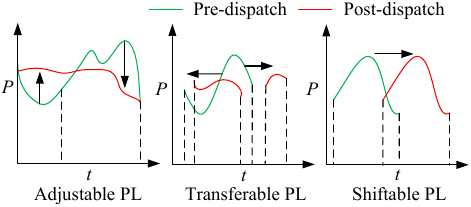}}
    \caption{Classification of dispatch modes of process loads.}
    \label{response_mode}
    \vspace{-0.3cm}
\end{figure}

\subsection{EAF--LF--CC Process Configuration}
\label{Process Configuration}

The considered multiline short-process steelmaking plant follows the EAF--LF--CC route. Scrap is first melted in an EAF, subsequently refined in an LF, and finally solidified through continuous casting. Each production batch, hereafter referred to as a heat, must pass through these stages sequentially. The corresponding SPLs comprise the electrical demands of the processing stages along this route.

Let $\mathcal L=\{1,2,\ldots,L\}$ denote the set of production lines, and let $\mathcal K_{\ell}$ denote the set of heats assigned to line $\ell$. Each line is equipped with an EAF, an LF, and a CC, whose device set is denoted by $\mathcal S_{\ell}$. The complete device and heat sets are given by:
\begin{equation}
\mathcal S=
\bigcup_{\ell\in\mathcal L}
\mathcal S_{\ell},
\quad
\mathcal K=
\bigcup_{\ell\in\mathcal L}
\mathcal K_{\ell},
\end{equation}

\noindent where $\mathcal S_{\ell}=\{s_{\ell}^{\mathrm{EAF}},
s_{\ell}^{\mathrm{LF}},s_{\ell}^{\mathrm{CC}}\}$. Let $\mathcal A\subseteq\mathcal S\times\mathcal S$ denote the set of adjacent upstream-downstream device pairs, where $(s,r)\in\mathcal A$ indicates that device $r$ immediately follows device $s$ in the production route.

\subsection{Processing Dynamics and Operational Constraints of SPLs}
\label{SPL Model}

The operation of SPLs is characterized by the coupling between device power consumption, batch-processing progress, and interstage production requirements. The corresponding processing-energy dynamics, device-level operating constraints, and interstage timing constraints are formulated below.

\subsubsection{Processing Dynamics}

Let $\delta_{k,t}^{\ell,s}\in\{0,1\}$ indicate whether device $s$ is processing heat $k$ on line $\ell$ during interval $t$, and let $P_{k,t}^{\ell,s}$ denote the corresponding power consumption. For each processing stage, $E_{k,t}^{\ell,s}$ denotes the cumulative electrical energy delivered to heat $k$ by the end of interval $t$, with $E_{k,0}^{\ell,s}=0$. The state update equation is given by:
\begin{equation}
\begin{aligned}
E_{k,t}^{\ell,s}
=E_{k,t-1}^{\ell,s}+P_{k,t}^{\ell,s}\Delta t, \quad\forall \ell\in\mathcal L,
k\in\mathcal K_{\ell}, s\in\mathcal S_{\ell}, t\in\mathcal T,
\end{aligned}
\end{equation}

\noindent where $\Delta t$ is the 5-min dispatch interval, and $\mathcal T$ denotes the complete set of dispatch intervals. The processing-energy state satisfies $0\le E_{k,t}^{\ell,s}\le\overline E_{k}^{\ell,s}$, where $\overline E_{k}^{\ell,s}$ is the electrical energy required to complete the corresponding processing stage.

\subsubsection{Device Operation Constraints}

The operating power of each device is bounded by:
\begin{equation}
\delta_{k,t}^{\ell,s}\underline P^{s}
\le
P_{k,t}^{\ell,s}
\le
\delta_{k,t}^{\ell,s}\overline P^{s},
\quad
\forall \ell\in\mathcal L,\ 
k\in\mathcal K_{\ell},\ 
s\in\mathcal S_{\ell},\ 
t\in\mathcal T,
\end{equation}

\noindent where $\underline P^{s}$ and $\overline P^{s}$ denote the minimum and maximum operating powers of device $s$, respectively. For a fixed-power device, $\underline P^{s}=\overline P^{s}$. Each device processes at most one heat during each dispatch interval:
\begin{equation}
\sum_{k\in\mathcal K_{\ell}}
\delta_{k,t}^{\ell,s}
\le
1,
\quad
\forall \ell\in\mathcal L,\ 
s\in\mathcal S_{\ell},\ 
t\in\mathcal T.
\end{equation}

For heat $k$ processed on device $s$, let
$t_{k,\mathrm{st}}^{\ell,s}$ and
$t_{k,\mathrm{fn}}^{\ell,s}$ denote its start and finish intervals, respectively. Each processing task occupies one continuous and noninterruptible interval:
\begin{subequations}
\begin{gather}
t_{k,\mathrm{st}}^{\ell,s}-M\left(1-\delta_{k,t}^{\ell,s}\right)
\le t \le t_{k,\mathrm{fn}}^{\ell,s}+M\left(1-\delta_{k,t}^{\ell,s}\right),
\\
\sum_{t\in\mathcal T}\delta_{k,t}^{\ell,s}=t_{k,\mathrm{fn}}^{\ell,s}
-t_{k,\mathrm{st}}^{\ell,s}+1, \quad
1\le t_{k,\mathrm{st}}^{\ell,s}\le t_{k,\mathrm{fn}}^{\ell,s}\le
|\mathcal T|,
\end{gather}
\end{subequations}

\noindent $t_{k,\mathrm{st}}^{\ell,s},\ t_{k,\mathrm{fn}}^{\ell,s}\in
\mathbb Z_{+}$, $\ell\in\mathcal L$, $k\in\mathcal K_{\ell}$,
$s\in\mathcal S_{\ell}$, and $t\in\mathcal T$, where
$M\ge|\mathcal T|$ is a sufficiently large constant.

\subsubsection{Interstage Timing and Production Requirements}

For each adjacent device pair $(s,r)\in\mathcal A$, heat $k$ can enter device $r$ only after completing its processing on device $s$. The interstage transfer and waiting time is constrained by:
\begin{equation}
\underline{\tau}_{\mathrm{tr}}^{s,r}
\le
\left(
t_{k,\mathrm{st}}^{\ell,r}
-
t_{k,\mathrm{fn}}^{\ell,s}
-
1
\right)\Delta t
\le
\overline{\tau}_{\mathrm{tr}}^{s,r},
\quad
\forall \ell\in\mathcal L,\ k\in\mathcal K_{\ell},
\end{equation}

\noindent where $\underline{\tau}_{\mathrm{tr}}^{s,r}$ and
$\overline{\tau}_{\mathrm{tr}}^{s,r}$ denote the minimum transfer time and the maximum admissible interval between adjacent processing stages, respectively. The processing duration of each heat is restricted by:
\begin{equation}
\underline{\tau}_{k}^{\ell,s}\le\left(t_{k,\mathrm{fn}}^{\ell,s}
-t_{k,\mathrm{st}}^{\ell,s}+1\right)\Delta t \le
\overline{\tau}_{k}^{\ell,s},\
\forall \ell\in\mathcal L,
k\in\mathcal K_{\ell},
s\in\mathcal S_{\ell},
\end{equation}

\noindent where $\underline{\tau}_{k}^{\ell,s}$ and
$\overline{\tau}_{k}^{\ell,s}$ denote the minimum and maximum admissible processing durations, respectively. A heat is regarded as completed when its CC processing-energy
requirement is satisfied. Accordingly, the number of completed heats
at the end of the dispatch horizon $N_{|\mathcal T|}^{\mathrm{fin}}$ is defined as:
\begin{equation}
N_{|\mathcal T|}^{\mathrm{fin}}
=
\sum_{\ell\in\mathcal L}
\sum_{k\in\mathcal K_{\ell}}
\mathbb I
\left\{
E_{k,|\mathcal T|}^{\ell,s_{\ell}^{\mathrm{CC}}}
=
\overline E_{k}^{\ell,s_{\ell}^{\mathrm{CC}}}
\right\}.
\end{equation}

\begin{figure}[htbp] 
    \vspace{-0.1cm}
    \centerline{\includegraphics[width=1\columnwidth]{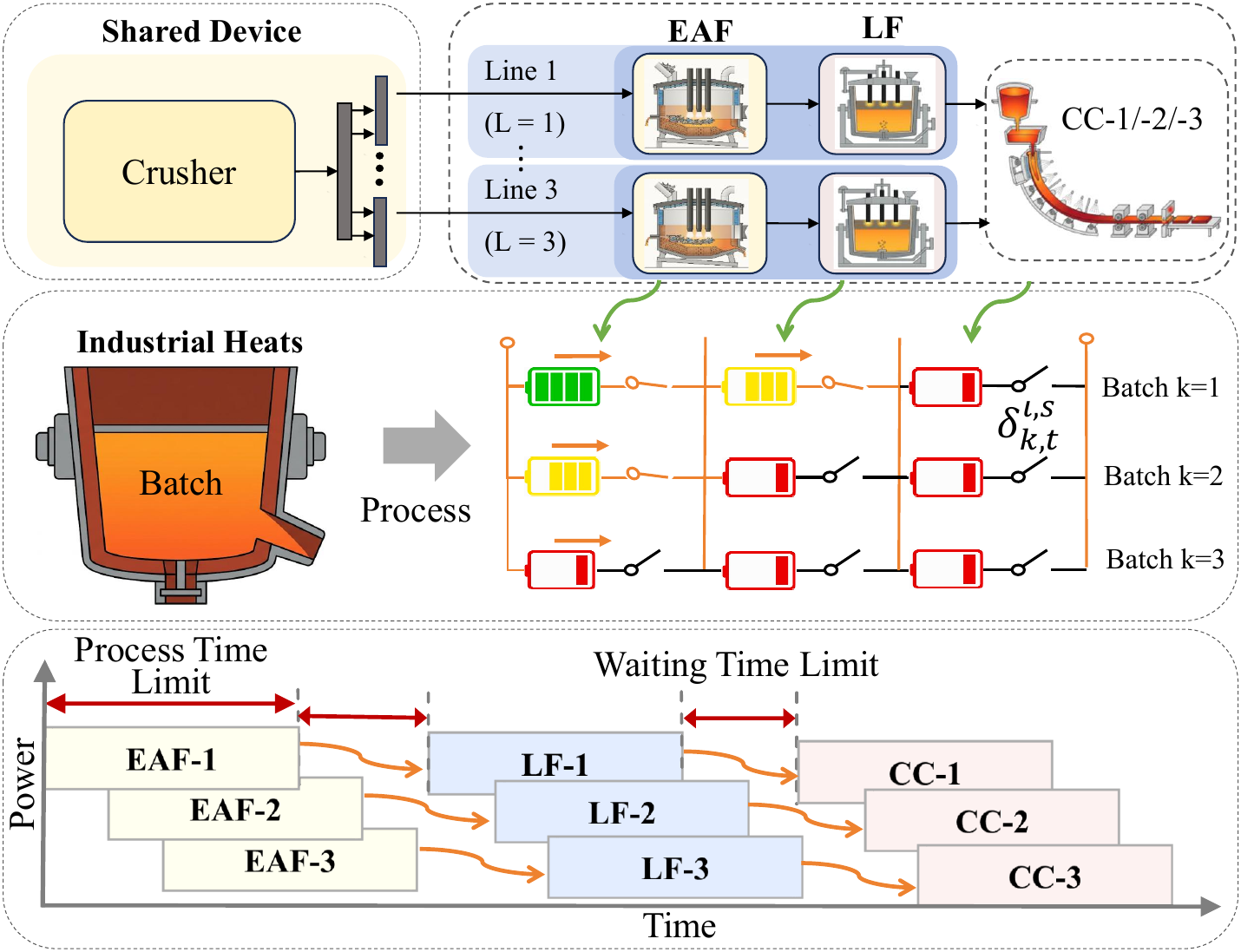}}
    \caption{Modeling diagram for steelmaking plants.}
    \label{VES_model}
    \vspace{-0.5cm}
\end{figure}

\section{PK-SDRL Dispatch Framework}\label{method}

This section develops the process-knowledge-embedded safe deep
reinforcement learning (PK-SDRL) framework for real-time dispatch
of SPLs under a two-part tariff with real-time electricity prices
and local renewable generation. At each 5-min decision interval, the dispatch objective is to reduce electricity procurement cost and demand-exceedance risk while satisfying process constraints. The expected process-correction distance induced by PDG-AP is further incorporated into policy updating to guide the raw actor toward the process-feasible action support. The overall framework of the proposed PK-SDRL framework is illustrated in Fig.~\ref{Flowchart}.

\subsection{Markov Decision Process Formulation} 

The intraday dispatch problem is modeled as a Markov decision process:
\begin{equation}
\mathcal M=(\mathcal X,\mathcal U,\mathbb P,r,\gamma),
\end{equation}
where $\mathcal X$, $\mathcal U$, $\mathbb P$, $r$, and $\gamma$ denote the state space, action space, transition kernel, instantaneous reward, and discount factor, respectively. Each episode corresponds to one operating day, and each decision step represents a 5-min interval. The environment state at time $t$ is constructed from the previous plant state $\Xi_{t-1}=[\{\delta_{k,t-1}^{\ell,s}\},\{P_{k,t-1}^{\ell,s}\}, \{E_{k,t-1}^{\ell,s}\}]$ and short-term forecast information, and a compact state representation is adopted:
\begin{equation}
\mathbf s_t=
\Big[
\Xi_{t-1},\,
\{\hat\lambda_{\tau|t},\hat P_{\tau|t}^{\mathrm{R}}\}_{\tau=t}^{t+H_s-1},\,
\boldsymbol{\eta}_t
\Big],
\end{equation}

\noindent where $\hat\lambda_{\tau|t}$ and $\hat P_{\tau|t}^{\mathrm{R}}$ denote the electricity-price and renewable-generation forecasts over the short look-ahead horizon $H_s$, respectively. The vector $\boldsymbol{\eta}_t$ collects the remaining daily production quota, the remaining time and the headroom relative to the contract demand $D$.

\subsection{Active Frontier and Action Space Construction}

To reduce the action-space dimension, decisions are restricted to the active-frontier heats. For each device $s\in\mathcal S_\ell$, the active frontier is defined as:
\begin{equation}
\mathcal F_t^{\ell,s}
=
\{k_t^{\ell,s,\mathrm{cur}}\}
\cup
\{k_t^{\ell,s,\mathrm{next}}\}, \
\forall \ell\in\mathcal L,\ 
s\in\mathcal S_{\ell},\ 
t\in\mathcal T,
\end{equation}
where $k_t^{\ell,s,\mathrm{cur}}$ and $k_t^{\ell,s,\mathrm{next}}$ denote the heat currently processed and the next eligible heat for device $s$ at time $t$, respectively. Eligibility is determined by the process constraints in Section~\ref{SPL Model}. The hybrid action is defined as:
\begin{equation}
a_t=\left(a_t^x,\mathbf P_t\right),
\end{equation}

\noindent where $a_t^x=\{\delta_{k,t}^{\ell,s}\},\ \mathbf P_t=\{P_{k,t}^{\ell,s}\},\ {k\in\mathcal F_t^{\ell,s}}$ collect the operating decisions and power setpoints, respectively, for the current and next heats on each device.

\begin{remark}
\emph{
Under the prescribed processing order, every feasible hybrid action is supported on the active frontier: $(\delta_{k,t}^{\ell,s},P_{k,t}^{\ell,s})=(0,0)$ for all $k\in\mathcal K_\ell\setminus\mathcal F_t^{\ell,s}$. Consequently, restricting decisions to the active frontier is lossless relative to full-batch action construction.}
\end{remark}

\subsection{Process-Distance-Guided Action Processing}
\label{subsec:pdgam}

Let $a,b\in\mathcal U^x$ denote candidate discrete connection actions. Given the current plant state $\Xi_t$, define the safe and excluded action sets as $\mathcal U_t^{\mathrm{s}}$ and $\mathcal U_t^{\mathrm{e}}=\mathcal U^x\setminus\mathcal U_t^{\mathrm{s}}$, respectively. A discrete action belongs to $\mathcal U_t^{\mathrm{s}}$ only if it satisfies constraints (4)--(7). Let $\pi_\theta^x(\cdot\mid\mathbf s_t)\in\Delta(\mathcal U^x)$ denote the raw discrete policy. Assuming $\mathcal U_t^{\mathrm{s}}\neq\varnothing$, its total probability over the safe set is $Z_{\theta,t}=\sum_{b\in\mathcal U_t^{\mathrm{s}}}\pi_\theta^x(b\mid\mathbf s_t)$, and the normalized safe reference policy is:
\begin{equation}
\bar\pi_{\theta,t}^x(b\mid\mathbf s_t)=\pi_\theta^x(b\mid\mathbf s_t)/Z_{\theta,t}, \quad  b\in\mathcal U_t^{\mathrm{s}}.
\end{equation}

A standard hard mask directly uses $\bar\pi_{\theta,t}^x(\cdot\mid\mathbf s_t)$ for sampling. Equivalently, the probability mass removed from every excluded action is allocated to safe actions according to the same reference distribution, without distinguishing which safe action constitutes a smaller process correction. To incorporate this distinction, let $\boldsymbol\delta_t(a)\in\{0,1\}^{n_t^x}$ collect the active-frontier operating states induced by action $a$. The normalized process distance between $a\in\mathcal U_t^{\mathrm{e}}$ and $b\in\mathcal U_t^{\mathrm{s}}$ is defined as:
\begin{equation}
d_t(a,b)=\|\boldsymbol\delta_t(a)-\boldsymbol\delta_t(b)\|_1/n_t^x,
\end{equation}

\noindent where $n_t^x$ is the dimension of $\boldsymbol\delta_t(a)$. Thus, $d_t(a,b)\in[0,1]$ measures the fraction of active-frontier operating decisions that must be changed when replacing $a$ with $b$. For each excluded action $a\in\mathcal U_t^{\mathrm{e}}$, its probability allocation weights are determined by:
\begin{equation}
\begin{aligned}
\omega_{\theta,t}^{\star}(\cdot\mid a)
=\arg\min_{q\in\Delta(\mathcal U_t^{\mathrm{s}})}
\Big\{
&\mathbb E_{b\sim q}[d_t(a,b)]\\
&+\tau_m D_{\mathrm{KL}}\!\left(
q\,\|\,\bar\pi_{\theta,t}^x(\cdot\mid\mathbf s_t)
\right)
\Big\},
\end{aligned}
\label{eq:pdgam_weight}
\end{equation}
\noindent where $\Delta(\mathcal U_t^{\mathrm{s}})$ is the probability simplex over $\mathcal U_t^{\mathrm{s}}$, $D_{\mathrm{KL}}(\cdot\|\cdot)$ denotes the Kullback--Leibler divergence, and $\tau_m>0$ is a tradeoff coefficient. The first term represents a process-near safe substitute, whereas the second prevents the allocation from departing excessively from the actor's safe-action preference. Since the raw softmax policy is positive on $\mathcal U_t^{\mathrm{s}}$, the objective in \eqref{eq:pdgam_weight} is strictly convex and admits the unique solution:
\begin{equation}
\omega_{\theta,t}^{\star}(b\mid a)
=\frac{
\bar\pi_{\theta,t}^{x}(b\mid\mathbf s_t)
\exp\!\left[-d_t(a,b)/\tau_m\right]
}{
\displaystyle\sum_{c\in\mathcal U_t^{\mathrm{s}}}
\bar\pi_{\theta,t}^{x}(c\mid\mathbf s_t)
\exp\!\left[-d_t(a,c)/\tau_m\right]
}, 
\label{eq:pdgam_closed_form}
\end{equation}

\noindent where $a\in\mathcal U_t^{\mathrm{e}}, b\in\mathcal U_t^{\mathrm{s}}$, the resulting safety-processed policy is expressed as:
\begin{equation}
\tilde\pi_{\theta}^{x}(b\mid\mathbf s_t)
=\pi_{\theta}^{x}(b\mid\mathbf s_t)
+\displaystyle\sum_{a\in\mathcal U_t^{\mathrm{e}}}
\pi_{\theta}^{x}(a\mid\mathbf s_t)
\omega_{\theta,t}^{\star}(b\mid a), \; b\in\mathcal U_t^{\mathrm{s}},
\label{eq:pdgam_policy}
\end{equation}

\noindent where $\tilde\pi_{\theta}^{x}(b\mid\mathbf s_t)=0, b\in\mathcal U_t^{\mathrm{e}}$. Since $\omega_{\theta,t}^{\star}(\cdot\mid a)$ is a probability distribution, \eqref{eq:pdgam_policy} preserves the total probability mass and assigns zero probability to all excluded actions. Using $\bar\pi_{\theta,t}^x(\cdot\mid\mathbf s_t)$ as a feasible point of \eqref{eq:pdgam_weight} gives:
\begin{equation}
\mathbb E_{b\sim\omega_{\theta,t}^{\star}(\cdot\mid a)}
[d_t(a,b)]
\le
\mathbb E_{b\sim\bar\pi_{\theta,t}^{x}(\cdot\mid\mathbf s_t)}
[d_t(a,b)],
\quad a\in\mathcal U_t^{\mathrm{e}}.
\label{eq:pdgam_distance_bound}
\end{equation}

\begin{prop}[Recursive process feasibility under PDG-AP]
\label{prop:pdgam_recursive_feasibility}
Suppose that $\Xi_0$ is process feasible. For every
$\ell\in\mathcal L$ and every pair of consecutive heats
$k,k+1\in\mathcal K_\ell$, the configured processing-time bounds
satisfy
$\overline{\tau}_{k}^{\ell,\mathrm{LF}}
\le
\underline{\tau}_{k+1}^{\ell,\mathrm{EAF}}$
and
$\overline{\tau}_{k}^{\ell,\mathrm{LF}}
+
\overline{\tau}_{k}^{\ell,\mathrm{CC}}
\le
\underline{\tau}_{k+1}^{\ell,\mathrm{EAF}}
+
\overline{\tau}_{k+1}^{\ell,\mathrm{LF}}$.
If
$a_t^x\sim\tilde{\pi}_{\theta}^{x}(\cdot\mid\mathbf s_t)$
and $\Xi_t=F\bigl(\Xi_{t-1}, a_t^x, \mathbf P_t^{\star} \bigr)$, then:
\begin{equation}
\Pr\!\left(
\left\{
a_t^x\in\mathcal U_t^{\mathrm{s}}(\Xi_{t-1})
\right\}
\cap
\left\{
\mathcal U_{t+1}^{\mathrm{s}}(\Xi_t)\neq\varnothing
\right\}
\,\middle|\,
\mathbf s_t
\right)
=1,
\quad
t>0.
\label{eq:pdgam_recursive_feasibility}
\end{equation}
\end{prop}

Proposition~\ref{prop:pdgam_recursive_feasibility} establishes the recursive process feasibility of PDG-AP under the prescribed processing-time conditions. Specifically, starting from a process-feasible state, each safety-processed action leads to a successor state with at least one admissible continuation, thereby preventing process deadlock. The proof of Proposition~\ref{prop:pdgam_recursive_feasibility} is provided in Appendix~A.

Although the safety-processed policy guarantees process-feasible execution, the raw policy may retain probability mass on the excluded action set. To quantify its reliance on safety processing, the expected process-correction distance induced by PDG-AP is defined as:
\begin{equation}
\begin{aligned}
\mathcal C_t(\theta)
={}&
\sum_{a\in\mathcal U_t^{\mathrm e}}
\pi_\theta^x(a\mid\mathbf s_t)
\sum_{b\in\mathcal U_t^{\mathrm s}}
\omega_{\theta,t}^{\star}(b\mid a)
d_t(a,b),
\end{aligned}
\label{eq:expected_process_correction}
\end{equation}

\noindent where $\mathcal C_t(\theta)=0$ when
$\mathcal U_t^{\mathrm e}=\varnothing$. The probability assigned by the raw policy to excluded actions is:
\begin{equation}
p_t^{\mathrm e}(\theta)
=
\sum_{a\in\mathcal U_t^{\mathrm e}}
\pi_\theta^x(a\mid\mathbf s_t).
\label{eq:raw_excluded_probability}
\end{equation}

For $\mathcal U_t^{\mathrm e}\neq\varnothing$, define the minimum nonzero process-correction distance as:
\begin{equation}
d_t^{\min}
=
\min_{\substack{
a\in\mathcal U_t^{\mathrm e}, b\in\mathcal U_t^{\mathrm s}
}}
d_t(a,b)>0.
\label{eq:minimum_process_distance}
\end{equation}

\begin{prop}[Raw-policy process-infeasibility bound]
\label{prop:raw_policy_bound}
Let $D_{\mathrm{TV}}(\cdot,\cdot)$ denote the total variation distance. For any state satisfying $\mathcal U_t^{\mathrm s}\neq\varnothing$, the raw and
safety-processed discrete policies satisfy:
\begin{equation}
D_{\mathrm{TV}}
\left(
\pi_\theta^x(\cdot\mid\mathbf s_t),
\tilde\pi_\theta^x(\cdot\mid\mathbf s_t)
\right)=p_t^{\mathrm e}(\theta)\le\frac{\mathcal C_t(\theta)}{d_t^{\min}}
\le n_t^x\mathcal C_t(\theta).
\label{eq:raw_policy_bound}
\end{equation}
\end{prop}

\begin{IEEEproof}
For $\mathcal U_t^{\mathrm e}=\varnothing$, all terms in
\eqref{eq:raw_policy_bound} are zero. Otherwise,
\eqref{eq:pdgam_policy} assigns zero probability to excluded actions and transfers their total probability mass
$p_t^{\mathrm e}(\theta)$ to the safe set. Therefore, the
$\ell_1$ distance between the raw and processed policies is
$2p_t^{\mathrm e}(\theta)$. Moreover,
$\sum_{b\in\mathcal U_t^{\mathrm s}}
\omega_{\theta,t}^{\star}(b\mid a)=1$ and
$d_t(a,b)\ge d_t^{\min}$ yield
$\mathcal C_t(\theta)\ge
d_t^{\min}p_t^{\mathrm e}(\theta)$.
Since distinct binary connection actions satisfy
$\|\boldsymbol\delta_t(a)-\boldsymbol\delta_t(b)\|_1\ge1$,
$d_t^{\min}\ge1/n_t^x$, which proves
Proposition~\ref{prop:raw_policy_bound}.
\end{IEEEproof}

\begin{remark}
\emph{
As $\tau_m\rightarrow\infty$, \eqref{eq:pdgam_closed_form} satisfies
$\omega_{\theta,t}^{\star}(\cdot\mid a)
\rightarrow
\bar\pi_{\theta,t}^{x}(\cdot\mid\mathbf s_t)$,
and \eqref{eq:pdgam_policy} reduces exactly to the standard masked policy. For a fixed decision state,
$\mathcal U_t^{\mathrm{s}}$ and $d_t$ are independent of $\theta$.
Consequently, both $\tilde\pi_{\theta}^{x}$ and
$\mathcal C_t(\theta)$ remain differentiable with respect to the actor parameters. The former is used for discrete-action sampling and the PPO likelihood ratio, while the latter guides raw-policy internalization.}
\end{remark}

The continuous branch decides the actual power of activated adjustable devices. Let $z_{k,t}^{\ell,s}\in\mathbb R$ denote the unbounded latent action sampled from the conditional continuous policy $\pi_\theta^p(\cdot\mid\mathbf s_t,a_t^x)$. The corresponding device power is obtained through the bounded mapping:
\begin{equation}
P_{k,t}^{\ell,s}
=\delta_{k,t}^{\ell,s}(a_t^x)
\Bigg[\underline P^s
+\frac{\overline P^s-\underline P^s}{2}
\times\left(1+\tanh z_{k,t}^{\ell,s}\right)\Bigg].
\label{eq:bounded_actual_power}
\end{equation}

\begin{figure}[htbp] 
    \vspace{-0.1cm}
    \centerline{\includegraphics[width=1\columnwidth]{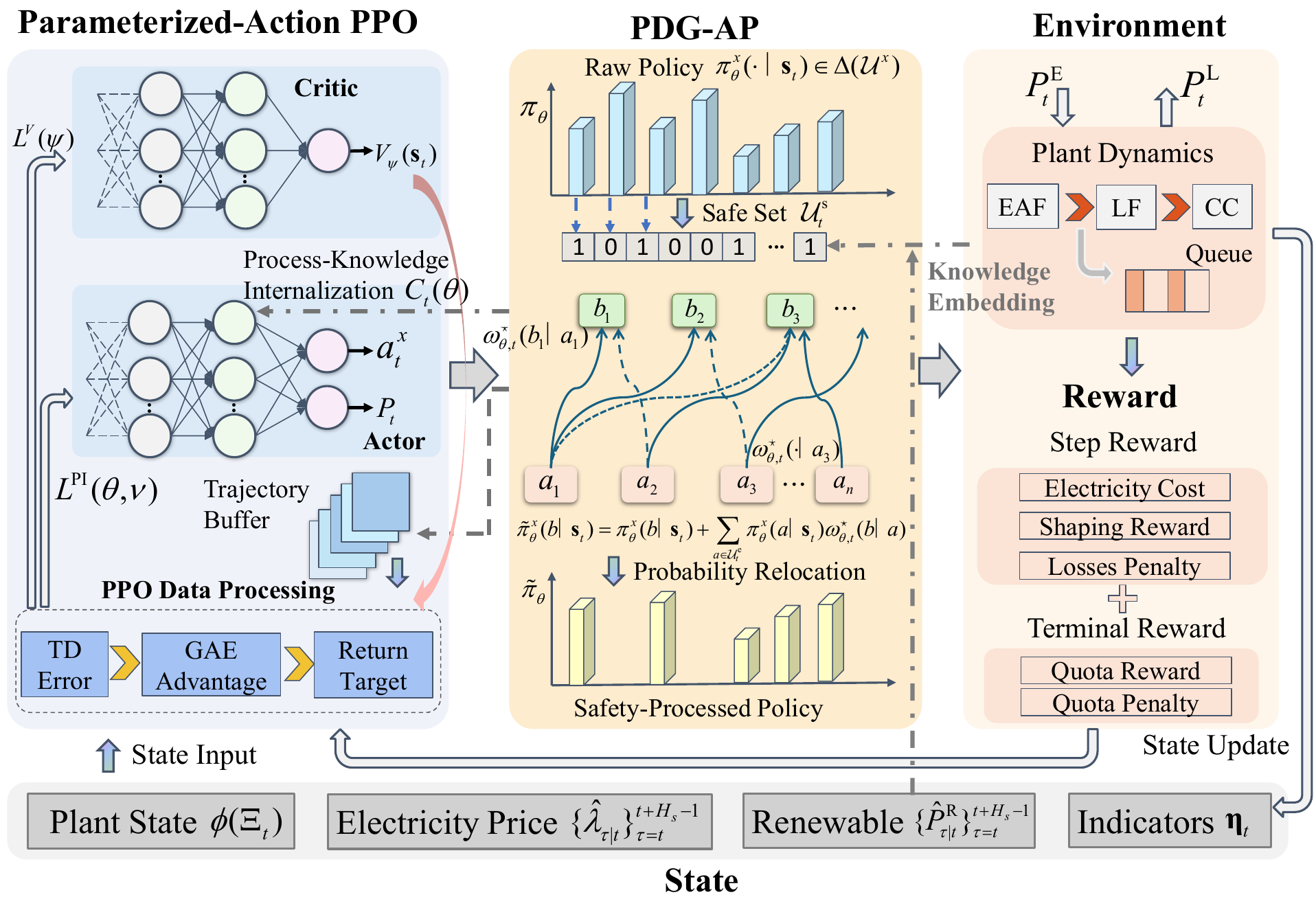}}
    \caption{PK-SDRL framework for real-time dispatch of SPLs.}
    \label{Flowchart}
    \vspace{-0.5cm}
\end{figure}

\subsection{Instantaneous Reward Function}

During intraday dispatch, local renewable generation is first used to supply the SPLs, while the remaining power demand is supplied by the external grid. Let $[x]^+=\max\{x,0\}$. The grid purchase power at time $t$ is given by:
\begin{equation}
P_t^{G}
=
\Bigg[
\sum_{\ell\in\mathcal L}
\sum_{s\in\mathcal S_{\ell}}
\sum_{k\in\mathcal K_{\ell}}
P_{k,t}^{\ell,s}
-
P_t^{\mathrm R}
\Bigg]^+,
\end{equation}
where $P_t^{\mathrm R}$ denotes the available local renewable generation. Accordingly, the renewable power utilized by the plant is
$P_t^{\mathrm{R,u}}
=
\min\{
P_t^{\mathrm R},
\sum_{\ell\in\mathcal L}
\sum_{s\in\mathcal S_{\ell}}
\sum_{k\in\mathcal K_{\ell}}
P_{k,t}^{\ell,s}
\}$.
The reward consists of an intraday reward for economic operation together with process advancement and a terminal reward for daily production completion. The intraday reward is defined as:
\begin{equation}
\begin{aligned}
r_t^{\mathrm{in}}
=
&-c_t^{\mathrm e}P_t^{G}\Delta t
-c^{\mathrm R}P_t^{\mathrm{R,u}}\Delta t
-c^{\mathrm{exc}}c_t^{\mathrm e}[P_t^{G}-D]^+\Delta t  \\
&-c^{\mathrm{hm}} n_t^{\mathrm{hm}}
-c^{\mathrm{sp}} n_t^{\mathrm{sp}}
+r_t^{\mathrm{shp}},
\end{aligned}
\label{eq:intraday_reward}
\end{equation}
where $c_t^{\mathrm e}$ is the real-time grid electricity price, $c^{\mathrm R}$ is the local renewable price, $D$ is the contract-demand threshold, and $c^{\mathrm{exc}}$ is the demand-exceedance cost coefficient. $n_t^{\mathrm{hm}}$ and $n_t^{\mathrm{sp}}$ denote the numbers of heats violating the admissible EAF$\rightarrow$LF and LF$\rightarrow$CC interstage intervals during time $t$, with penalty coefficients $c^{\mathrm{hm}}$ and $c^{\mathrm{sp}}$, respectively. The shaping reward $r_t^{\mathrm{shp}}$ provides intermediate feedback on multi-stage production progress.

Since an upstream action affects final heat completion only after the subsequent processing stages, sparse terminal feedback may result in delayed credit assignment. A state-transition-based shaping term is therefore introduced. The potential function is defined as:
\begin{equation}
\Phi(\Xi_t)
=
\sum_{\ell\in\mathcal L}
\sum_{s\in\mathcal S_{\ell}}
\omega_s
\sum_{k\in\mathcal K_\ell}
\frac{E_{k,t}^{\ell,s}}
{\overline E_{k}^{\ell,s}},
\end{equation}
where $\omega_s\ge 0$ is the weight associated with processing stage $s$. The shaping reward with shaping coefficient $\eta\ge 0$ and discount factor $\gamma$ is then given by:
\begin{equation}
r_t^{\mathrm{shp}}=\eta\big(\gamma\Phi(\Xi_{t+1})-\Phi(\Xi_t)\big).
\end{equation}

At the end of each operating day, a terminal reward is applied to evaluate the daily quota fulfillment rate. Let $N_{|\mathcal T|}^{\mathrm{fin}}$, $N^{\mathrm{qta}}$, $\rho^{\mathrm{day}}$ denote the number of completed heats, daily quota and the reward for quota fulfillment, respectively. The terminal reward is defined as:
\begin{equation}
r_{|\mathcal T|}^{\mathrm{ter}}
=\rho^{\mathrm{day}}
\mathbb I\!\left\{N_{|\mathcal T|}^{\mathrm{fin}}\ge N^{\mathrm{qta}}\right\}
-c^{\mathrm{q}}
\big[N^{\mathrm{qta}}-N_{|\mathcal T|}^{\mathrm{fin}}\big]^+,
\label{eq:terminal_reward}
\end{equation}

\noindent where $c^{\mathrm{q}}$ is the unit production-shortfall penalty. The total discounted episodic reward used for training is therefore derived as:
\begin{equation}
R
=
\sum_{t=0}^{{|\mathcal T|}-1}\gamma^t r_t^{\mathrm{in}}
+
\gamma^{|\mathcal T|} r_{|\mathcal T|}^{\mathrm{ter}}.
\label{eq:total_episode_reward}
\end{equation}

\subsection{Parameterized-Action PPO with Process-Knowledge Internalization}
\label{subsec:parameterized_action_ppo}

Since the action comprises a discrete connection decision $a_t^x$ and continuous power parameters, a parameterized-action actor--critic framework is adopted. The actor network contains two output branches. The discrete branch generates the raw probabilities of the candidate connection actions, which are transformed into the safety-processed policy $\tilde{\pi}_{\theta}^{x}(a_t^x\mid\mathbf s_t)$ through the process-knowledge embedding described above. After $a_t^x$ is sampled from this policy, the continuous branch generates a conditional policy for the latent power parameters:
\begin{equation}
\pi_{\theta}^{p}
\left(
\mathbf z_t^{p}
\mid \mathbf s_t,a_t^x
\right)
=
\mathcal N
\left(
\boldsymbol{\mu}_{\theta}^{p}(\mathbf s_t,a_t^x),
\boldsymbol{\Sigma}_{\theta}^{p}
\right),
\end{equation}
where $\mathbf z_t^{p}$ contains the latent power parameters associated with the operations activated by $a_t^x$, $\boldsymbol{\mu}_{\theta}^{p}(\cdot)$ denotes the corresponding mean output, and
$\boldsymbol{\Sigma}_{\theta}^{p}
=\operatorname{diag}((\boldsymbol{\sigma}_{\theta}^{p})^2)$
is the diagonal exploration covariance. The sampled latent parameters are converted into executable device powers through the deterministic bounded-power mapping defined above:
\begin{equation}
\mathbf P_t^{\star}
=
\mathcal G_t^{p}
\left(
\mathbf z_t^{p};
\mathbf s_t,a_t^x
\right).
\end{equation}

Accordingly, letting
$\mathbf z_t=(a_t^x,\mathbf z_t^{p})$
denote the parameterized action sampled by the policy, its joint distribution is written as:
\begin{equation}
\tilde{\pi}_{\theta}
\left(
\mathbf z_t\mid\mathbf s_t
\right)
=
\tilde{\pi}_{\theta}^{x}
\left(
a_t^x\mid\mathbf s_t
\right)
\pi_{\theta}^{p}
\left(
\mathbf z_t^{p}
\mid\mathbf s_t,a_t^x
\right).
\end{equation}

The joint likelihood of $\mathbf z_t$ is used for PPO updating, whereas the corresponding processed action
$a_t^x, \mathbf P_t^{\star}$
is executed in the production environment. Let $V_{\psi}(\mathbf s_t)$ denote the critic network parameterized by $\psi$, and let $\psi_{\mathrm{old}}$ denote its parameters when the trajectory is collected. Let
$r_t \triangleq r_t^{\mathrm{in}}
+\mathbb{I}_{\{t=|\mathcal T|\}}r_{|\mathcal T|}^{\mathrm{ter}}$
denote the one-step reward used for policy optimization. Given the reward $r_t$, the one-step temporal-difference error is:
\begin{equation}
\delta_t
=
r_t
+
\gamma
V_{\psi_{\mathrm{old}}}(\mathbf s_{t+1})
-
V_{\psi_{\mathrm{old}}}(\mathbf s_t),
\end{equation}

The terminal-state value is set to zero, and the  generalized advantage estimation (GAE) is computed as:
\begin{equation}
\hat A_t
=
\sum_{j=0}^{J_t-1}
(\gamma\lambda)^j\delta_{t+j},
\end{equation}
where $\lambda\in[0,1]$ is the GAE trace-decay parameter controlling the bias--variance tradeoff, and $J_t$ denotes the number of valid steps from $t$ to the end of the collected rollout or episode. The return target for critic learning is given by:
\begin{equation}
\hat R_t
=
\hat A_t
+
V_{\psi_{\mathrm{old}}}(\mathbf s_t).
\end{equation}

Both $\hat A_t$ and $\hat R_t$ are held fixed during each PPO update. A clipped policy update is employed to improve training stability. Let $\theta_{\mathrm{old}}$ denote the actor parameters used to collect the trajectory. The probability ratio is computed from the joint likelihood of the sampled safe connection action and latent power parameters:
\begin{equation}
\rho_t(\theta)
=
\frac{
\tilde{\pi}_{\theta}^{x}(a_t^x\mid\mathbf s_t)
\pi_{\theta}^{p}
\left(
\mathbf z_t^{p}
\mid\mathbf s_t,a_t^x
\right)
}{
\tilde{\pi}_{\theta_{\mathrm{old}}}^{x}(a_t^x\mid\mathbf s_t)
\pi_{\theta_{\mathrm{old}}}^{p}
\left(
\mathbf z_t^{p}
\mid\mathbf s_t,a_t^x
\right)
}.
\end{equation}

The clipped actor objective with clipping parameter $\epsilon$ is written as:
\begin{equation}
L^{\mathrm{clip}}(\theta)
=
\mathbb E_t
\left[
\min
\left(
\rho_t(\theta)\hat A_t,\,
\mathrm{clip}
\left(
\rho_t(\theta),1-\epsilon,1+\epsilon
\right)
\hat A_t
\right)
\right].
\end{equation}

The actor update is further constrained by the expected process-correction distance:
\begin{equation}
\max_{\theta}\quad
L^{\mathrm{clip}}(\theta), \quad
\mathrm{s.t.}\quad
\mathbb E_t\left[\mathcal C_t(\theta)\right]
\le\kappa,
\label{eq:correction_budgeted_ppo}
\end{equation}

\noindent where $\kappa>0$ denotes the prescribed process-correction budget. The corresponding primal--dual actor objective is:
\begin{equation}
L^{\mathrm{PI}}(\theta,\nu)
={}
L^{\mathrm{clip}}(\theta)-
\nu
\left(
\mathbb E_t[\mathcal C_t(\theta)]
-\kappa
\right),
\quad \nu\ge0,
\label{eq:internalized_actor_objective}
\end{equation}
where $\nu$ is the dual variable. The actor parameters are updated by maximizing
$L^{\mathrm{PI}}(\theta,\nu)$, while $\nu$ is updated according to:
\begin{equation}
\nu
\leftarrow
\left[
\nu
+
\eta_{\nu}
\left(
\mathbb E_t[\mathcal C_t(\theta)]
-\kappa
\right)
\right]_{+},
\label{eq:dual_update}
\end{equation}

\noindent where $\eta_{\nu}>0$ is the dual learning rate. Let
$\overline n^x=\max_{t\in\mathcal T}n_t^x$.
From Proposition~\ref{prop:raw_policy_bound}, any policy satisfying
\eqref{eq:correction_budgeted_ppo} obeys:
\begin{equation}
\mathbb E_t
\left[
p_t^{\mathrm e}(\theta)
\right]
\le
\overline n^x
\mathbb E_t
\left[
\mathcal C_t(\theta)
\right]
\le
\overline n^x\kappa.
\label{eq:expected_infeasible_probability_bound}
\end{equation}

Hence, PDG-AP guarantees process-feasible execution, while the correction-budgeted update limits the raw actor's dependence on safety processing. PDG-AP remains active during deployment to preserve the hard feasibility guarantee. The critic is updated by minimizing the value loss below, and the complete training procedure of the proposed PK-SDRL method is summarized in Algorithm~\ref{alg:pksdrl}.
\begin{equation}
L^{V}(\psi)
=\mathbb E_t
\left[
\left(
V_{\psi}(\mathbf s_t)-\hat R_t
\right)^2
\right].
\end{equation}

\begin{algorithm}[t]
\caption{Training Procedure for PK-SDRL}
\label{alg:pksdrl}
\SetAlgoLined
\SetEndCharOfAlgoLine{}

\KwIn{Training-day data $\mathcal D$, process model, and hyperparameters.}
\KwOut{Trained policy $\tilde{\pi}_{\theta^\ast}$.}

Initialize $\theta$, $\psi$, and $\nu\leftarrow0$.\;

\For{each episode drawn from $\mathcal D$}{
    Reset $\Xi_0$ and $\mathcal B\leftarrow\varnothing$, and set
    $(\theta_{\mathrm{old}},\psi_{\mathrm{old}})
    \leftarrow(\theta,\psi)$.\;

    \For{$t\in\mathcal T$}{
        Observe $\mathbf s_t$ and construct
        $\{\mathcal F_t^{\ell,s}\}$,
        $\mathcal U_t^{\mathrm s}$, and $\mathcal U_t^{\mathrm e}$.\;

        Obtain $\tilde{\pi}_{\theta_{\mathrm{old}}}^{x}$
        through PDG-AP using
        \eqref{eq:pdgam_closed_form}--\eqref{eq:pdgam_policy}.\;

        Sample
        $\mathbf z_t=(a_t^x,\mathbf z_t^p)
        \sim\tilde{\pi}_{\theta_{\mathrm{old}}}
        (\cdot\mid\mathbf s_t)$.\;

        Map $\mathbf z_t^p$ to $\mathbf P_t^\star$, execute the action,
        and observe $r_t$ and $\mathbf s_{t+1}$.\;

        Store $(\mathbf s_t,\mathbf z_t,r_t,\mathbf s_{t+1})$
        in $\mathcal B$.\;
    }

    Compute $\hat A_t$ and $\hat R_t$ from $\mathcal B$ using GAE.\;

    Update $\theta$ by maximizing $L^{\mathrm{PI}}$, update $\psi$
    by minimizing $L^V$, and update $\nu$ using
    \eqref{eq:dual_update}.\; 
}
\end{algorithm}

\section{Case Study}\label{Case Study}

\subsection{Set up}

\subsubsection{Renewable Energy Sources and Market Data}
Real-time wind and photovoltaic generation data were collected from a renewable energy station in an industrial microgrid, where the installed wind and PV capacities are 425 MW and 375 MW, respectively. The dataset spans 14 months with a one-second resolution and covers a broad range of typical renewable generation scenarios from Oct/2024-Nov/2025, as shown in Fig.~\ref{RES}.

\begin{figure}[htbp] 
    \vspace{-0.1cm}
    \centerline{\includegraphics[width=1\columnwidth]{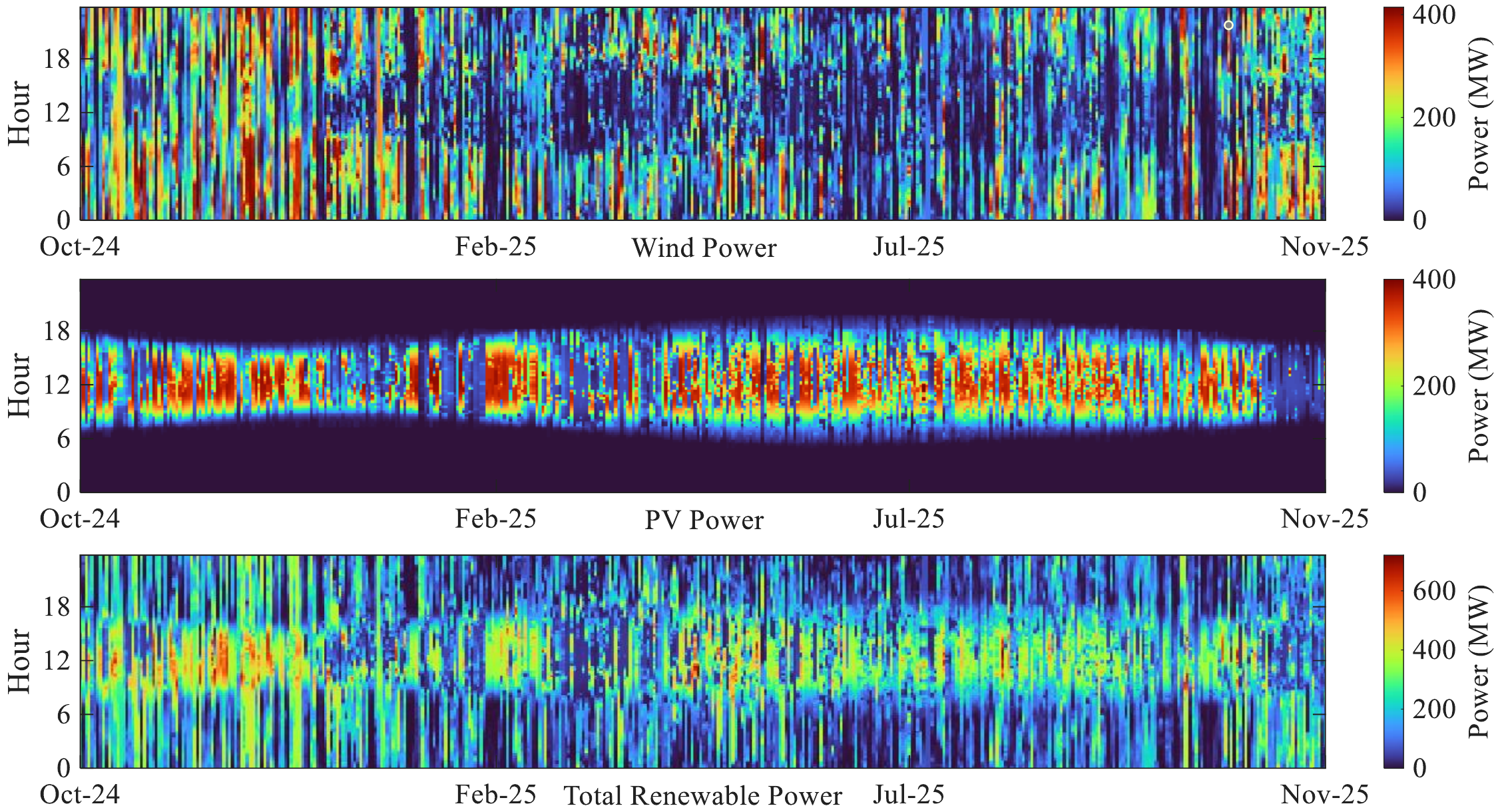}}
    \caption{Distribution and variation trends of wind and photovoltaic power output over a 14-month period.}
    \label{RES}
    \vspace{-0.1cm}
\end{figure}

Real-time system electricity price data were collected from PJM at a 5-minute resolution and were used for RL training and testing. The daily price profiles are illustrated in Fig.~\ref{Price}.

\begin{figure}[htbp] 
    \vspace{-0.1cm}
    \centerline{\includegraphics[width=0.9\columnwidth]{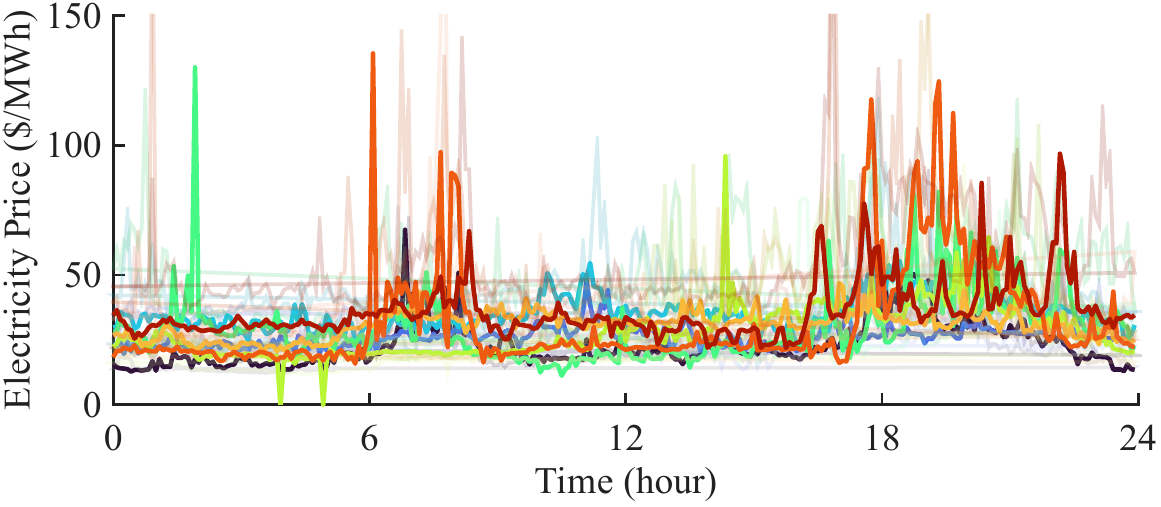}}
    \caption{Typical 5-minute real-time system electricity price profiles for PJM from Oct/2024 to Nov/2025.}
    \label{Price}
    \vspace{-0.5cm}
\end{figure}

\subsubsection{System Configuration and Parameters}

A test system is constructed based on a representative large-scale short-process steel plant. The system includes one crushing unit and three steelmaking lines,
each of which consists of one EAF, one LF, and one CC. The main equipment parameters are summarized in Table~\ref{tab:line_parameters}, and the rated power of each unit is specified based on an 80-ton heat size per batch. Both unit commitment and power dispatch are implemented at a 5-minute resolution, aligned with the electricity price data.

\subsubsection{Simulation Environment}

RL training and testing are implemented in MATLAB (R2025b) on a workstation equipped with an Intel Core Ultra 9 275HX processor at 2.70 GHz and 32 GB RAM. Gurobi, interfaced through YALMIP, is used to solve the optimization problems. During RL training, the agent is trained for 200 epochs by repeatedly cycling through 400 days of historical data, while the validation set consists of an additional 15 days of unseen renewable generation and electricity price data. The state variables include electricity-price and renewable-generation forecasts over a 3-h look-ahead horizon. This horizon is consistent with the practical range of ultra-short-term forecasting at a 5-min resolution~\cite{forecast}. In testing, forecast data are generated by adding zero-mean Gaussian
errors to the realized electricity-price and renewable-generation trajectories,
with the error standard deviation $\sigma_{\mathrm f}$ set to 10\% of the realized values.

\vspace{-0.5cm}
\begin{table}[htbp]
\caption{Key Device and Parameters of the Steel Plant}
\label{tab:line_parameters}
\centering
\footnotesize
\setlength{\tabcolsep}{2pt}
\renewcommand{\arraystretch}{1.08}
\resizebox{\columnwidth}{!}{%
\begin{tabular}{@{}lcccccc@{}}
\toprule
Parameter
& Crusher
& EAF
& LF 
& CC
& EAF$\rightarrow$LF
& LF$\rightarrow$CC \\
\midrule
$(\underline{P},\,\overline{P})$(MW) 
& (4,\,4)
& (45,\,75)
& (6,\,10)
& (4,\,4)
& \textbackslash 
& \textbackslash \\
Energy (MWh/heat) 
& \textbackslash 
& 34.8 
& 2.4 
& 2 
& \textbackslash 
& \textbackslash \\
Time (min) 
& \textbackslash 
& 40-50
& 20-30
& 30
& 5--10 
& 5--10 \\
Waiting time (min) 
& \textbackslash 
& $\ge 5 $
& $\ge 5 $
& $\ge 5 $
& $\leq 10$ 
& $\leq 10$ \\
\bottomrule
\end{tabular}%
}
\vspace{-0.5cm}
\end{table}

\subsection{Scheduling Results}\label{Scheduling Results}

\subsubsection{Training Convergence and Feasibility}\label{Training Convergence and Feasibility}

During training, the PK-SDRL agent is evaluated periodically to track its transition from stochastic exploration to deterministic operation. The policy is first trained with sampled actions to encourage exploration, while greedy validation is conducted at regular checkpoints to assess the performance of the deployable policy on unseen validation days. The key parameters are summarized in Table~\ref{tab:pksdrl_parameters}. Fig.~\ref{reward} presents the training convergence of the proposed PK-SDRL scheduling policy evaluated at greedy validation checkpoints. In the early training stage, the validation reward remains low and the quota fulfillment rate fluctuates noticeably, reflecting the agent's insufficient ability to meet the daily production target and the resulting terminal shortfall penalties. As training progresses, the validation reward increases rapidly and the quota fulfillment rate reaches one. This trend indicates that PK-SDRL first learns the primary feasibility requirement of completing 54 heats per day and then converges to a stable scheduling regime.

\vspace{-0.0cm}
\begin{table}[htbp]
\caption{Key Parameters Used in PK-SDRL Training}
\label{tab:pksdrl_parameters}
\centering
\footnotesize
\setlength{\tabcolsep}{1pt}
\renewcommand{\arraystretch}{1.08}
\resizebox{\columnwidth}{!}{%
\begin{tabular}{@{}cc@{\hspace{0.5em}}cc@{\hspace{0.5em}}cc@{}}
\toprule
Symbol & Value
& Symbol & Value
& Symbol & Value \\
\midrule

$N^{\mathrm{qta}}$
& $54$ heats/day
& $D$
& $100$ MW
& $H_s$
& 36
\\

$\sigma_{\mathrm f}$
& $10\%$
& $c^{\mathrm R}$
& $10$ \$/MWh
& $c^{\mathrm{exc}}$
& $2$
\\

$c^{\mathrm q}$
& $2.5\times10^{4}$ \$/heat
& $c^{\mathrm{hm}}$
& $1.4\times10^{4}$ \$/heat
& $c^{\mathrm{sp}}$
& $9.0\times10^{3}$ \$/heat
\\

$\rho^{\mathrm{day}}$
& $2.5\times10^{4}$ \$
& $\eta$
& $3.0\times10^{3}$ \$
& $\{\omega_s\}_{s\in\mathcal S}$
& $1/3$
\\

$\tau_m$
& $0.10$
& $\kappa$
& $0.05$
& $\eta_{\nu}$
& $1.0\times10^{-3}$
\\

\bottomrule
\end{tabular}%
}
\vspace{-0.2cm}
\end{table}

\begin{figure}[htbp] 
    \vspace{-0.1cm}
    \centerline{\includegraphics[width=0.95\columnwidth]{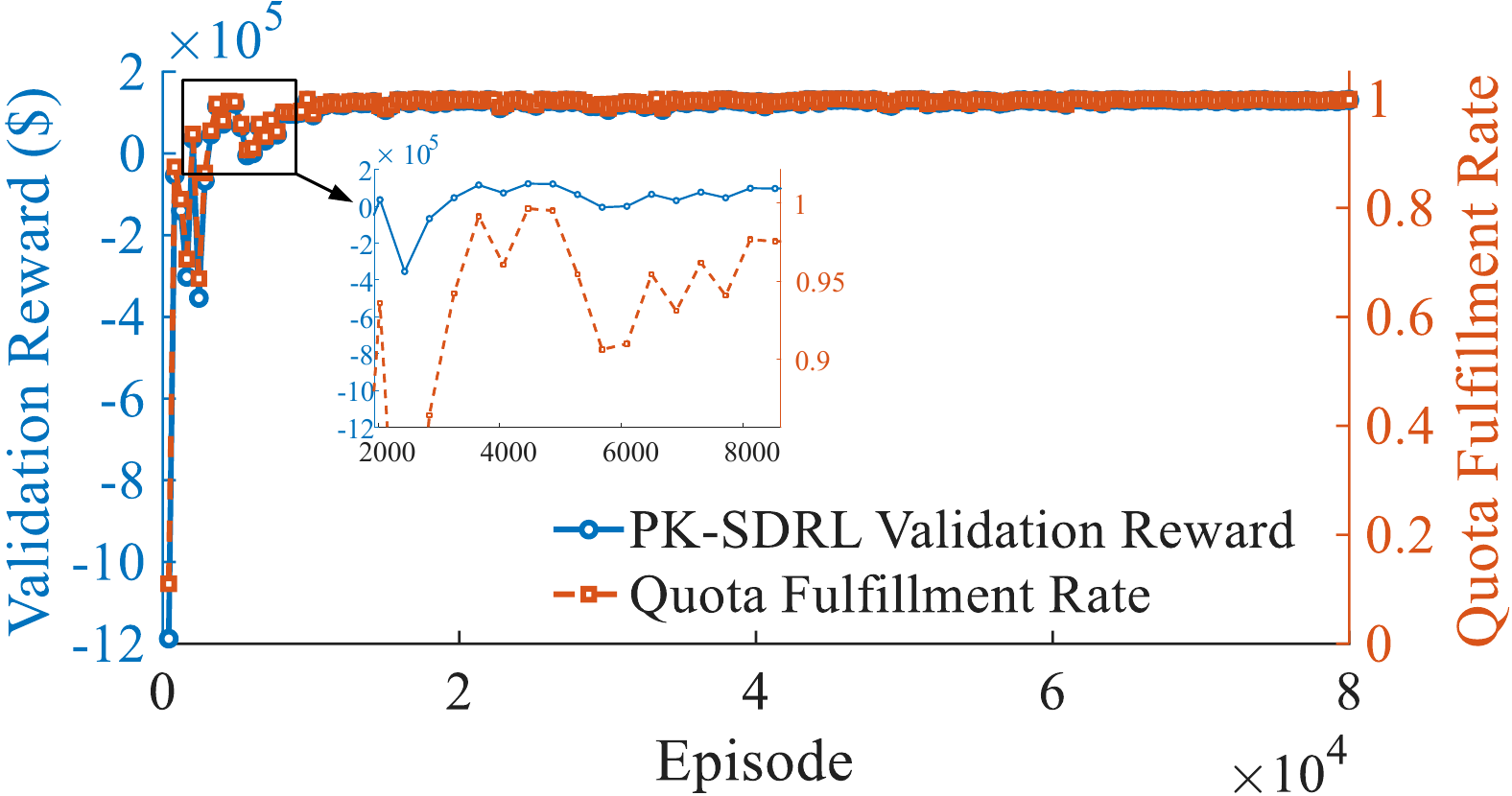}}
    \caption{Validation reward and quota fulfillment rate during the training of the proposed PK-SDRL scheduling policy.}
    \label{reward}
    \vspace{-0.1cm}
\end{figure}

\begin{figure}[htbp] 
    \vspace{-0.1cm}
    \centerline{\includegraphics[width=0.8\columnwidth]{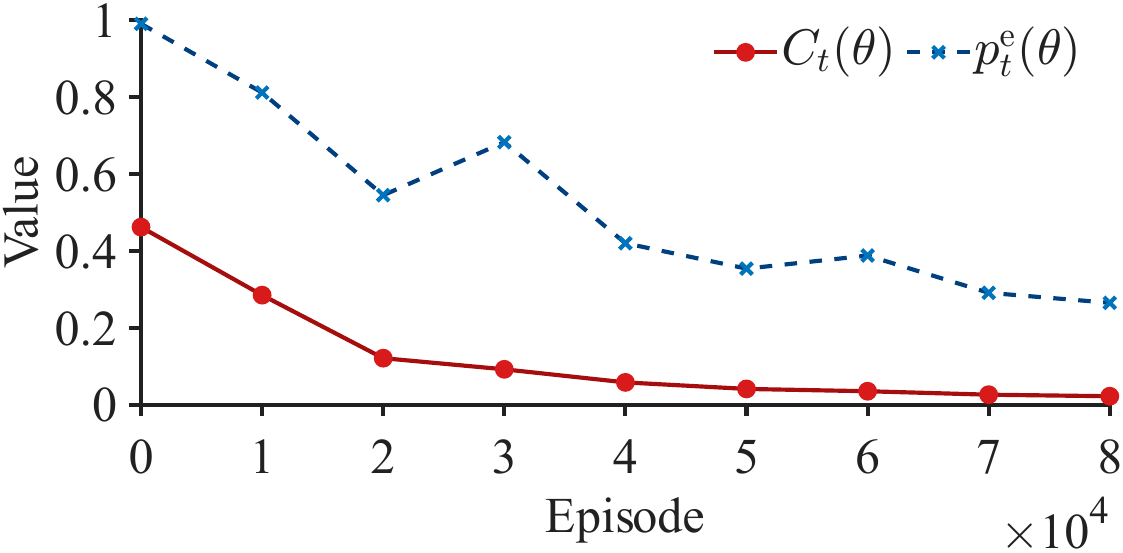}}
    \caption{Evolution of the raw excluded-action probability and expected process-correction distance during PK-SDRL training.}
    \label{PDG_AM}
    \vspace{-0.5cm}
\end{figure}

Fig.~\ref{PDG_AM} verifies the effectiveness of process-knowledge
internalization during policy training. As training proceeds, both the raw
excluded-action probability $p_t^{\mathrm e}(\theta)$ and the expected
process-correction distance $\mathcal C_t(\theta)$ decrease substantially,
indicating that the raw actor progressively shifts its probability mass toward
the safe action support. The faster reduction in $\mathcal C_t(\theta)$ further shows that the remaining excluded decisions require only minor process corrections, thereby reducing the policy's dependence on PDG-AP. Fig.~\ref{cost} further compares the selected PK-SDRL agent with a rule-based production policy on unseen validation days. The PK-SDRL policy achieves lower daily electricity cost across the validation set by adaptively coordinating production decisions with real-time electricity prices and renewable generation profiles. Meanwhile, both hot-metal losses and semi-product losses remain zero, demonstrating that the economic improvement is obtained while respecting the inter-stage transfer and waiting-time constraints.

\begin{figure}[htbp] 
    \vspace{-0.1cm}
    \centerline{\includegraphics[width=1\columnwidth]{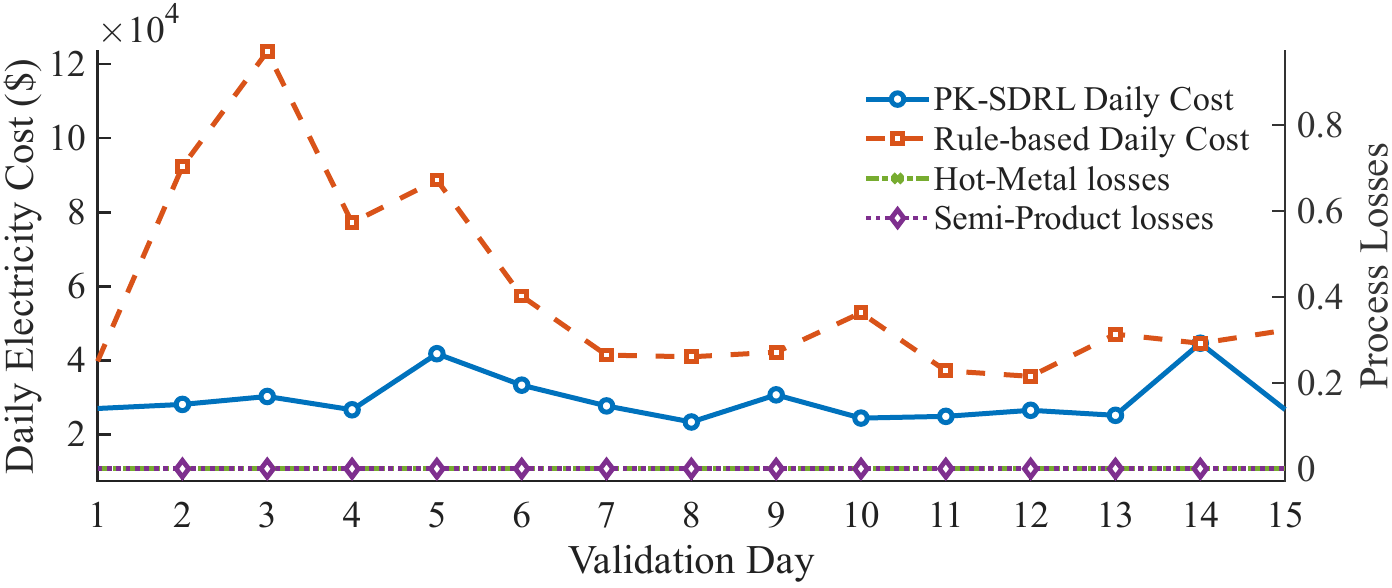}}
    \caption{Daily electricity cost comparison and process-loss performance on unseen validation days.}
    \label{cost}
    \vspace{-0.5cm}
\end{figure}

\subsubsection{Real-Time Decision Analysis of the PK-SDRL Framework}\label{Intra-Day Decision Analysis of the PK-SDRL Policy}

At the system level, PK-SDRL does not correspond to a deterministic heuristic that increases production at low-price periods and suspends production at high-price periods. Rather, it learns a dynamic decision-making policy that jointly accounts for the daily production requirement, the forecast-horizon grid electricity cost, and local renewable generation. For the representative validation day in Fig.~\ref{Fig_System_Level_Day_9}, renewable generation reaches nearly 400 MW, the electricity price peaks at approximately 70--80 \$/MWh, and the total plant power reaches about 220 MW. However, several high-load intervals coincide with high renewable generation, and therefore the grid import does not increase proportionally with the total plant power. This indicates that PK-SDRL exploits intervals with low grid-import pressure instead of merely suppressing load. When renewable generation is limited and the electricity price is unfavorable, the policy reduces production intensity. When renewable generation is sufficient or the price remains within an acceptable range, the policy increases production power to advance heat completion. The cumulative production reaches 54 heats before the end of the day, indicating that PK-SDRL leverages intra-day variations in price and renewable generation to reduce the effective grid electricity cost while satisfying the production requirement.

\begin{figure}[htbp] 
    \vspace{-0.1cm}
    \centerline{\includegraphics[width=1\columnwidth]{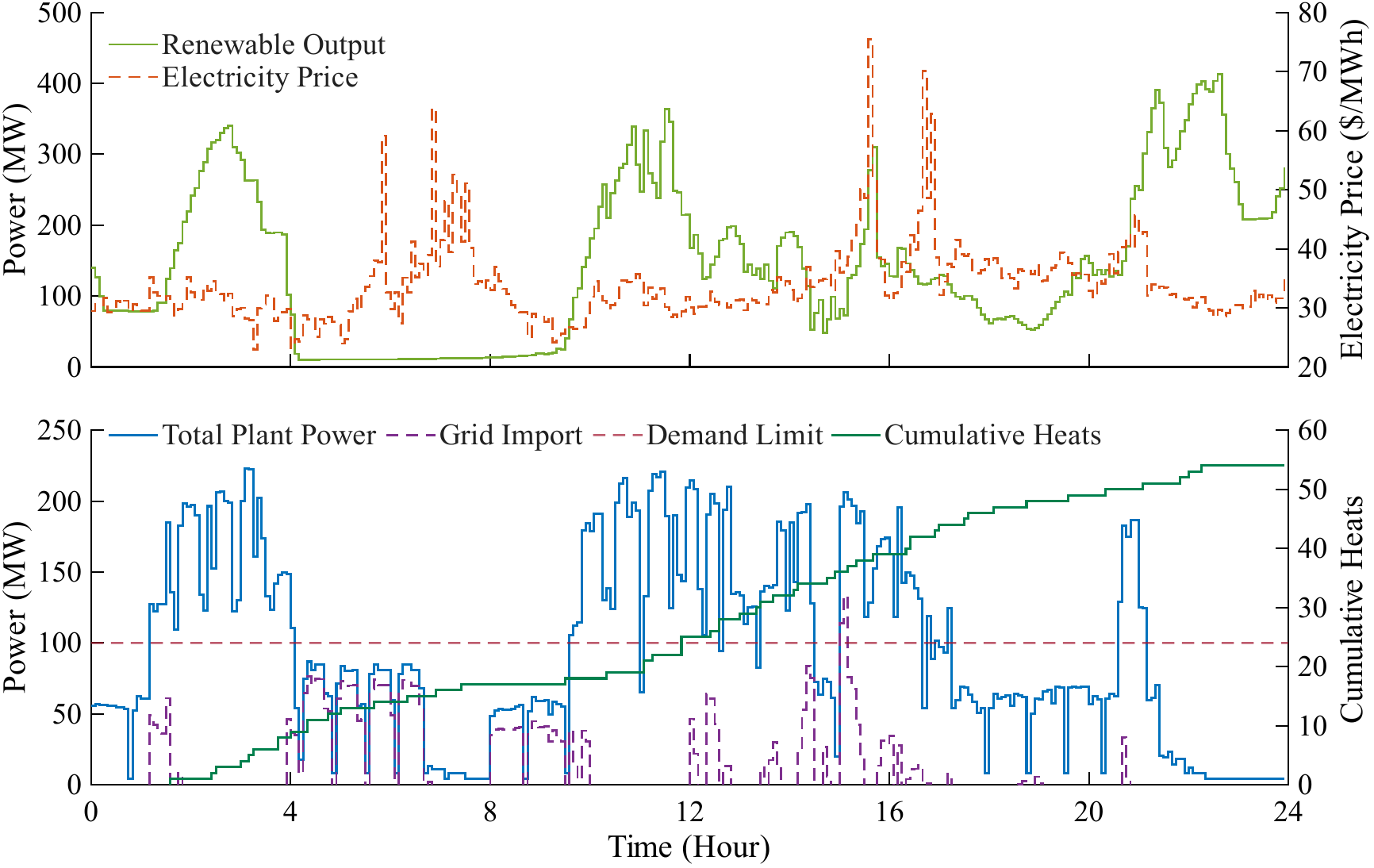}}
    \caption{Plant-level intra-day scheduling trajectory of the proposed PK-SDRL policy on a representative validation day.}
    \label{Fig_System_Level_Day_9}
    \vspace{-0.1cm}
\end{figure}

At the device level, PK-SDRL translates the heterogeneous flexibility of different process stages into an executable staggered production schedule. As shown in Fig.~\ref{Fig_Device_Level_Day_9}, EAFs serve as the primary high-power controllable loads, with operating powers varying around the 60 MW baseline, and therefore provide the main source of economic dispatch. By contrast, LFs and CCs operate at much lower power levels, approximately 8 MW and 4 MW, respectively, and their dispatching primarily support the continuity of upstream material transfer. The three production lines operate in a staggered manner instead of being started synchronously, which mitigates power peaks caused by simultaneous EAF operation and provides coordination space for subsequent LF and CC stages. The sequential EAF--LF--CC processing relationship is also maintained throughout the day, with no material-transfer infeasibility induced by price-driven scheduling. These results demonstrate that PK-SDRL does not simply learn a load-reduction policy. Under safety-aware action constraints, it establishes a hierarchical coordination pattern in which high-power stages provide economic flexibility, while downstream stages preserve process continuity.

\begin{figure}[htbp] 
    \vspace{-0.1cm}
    \centerline{\includegraphics[width=1\columnwidth]{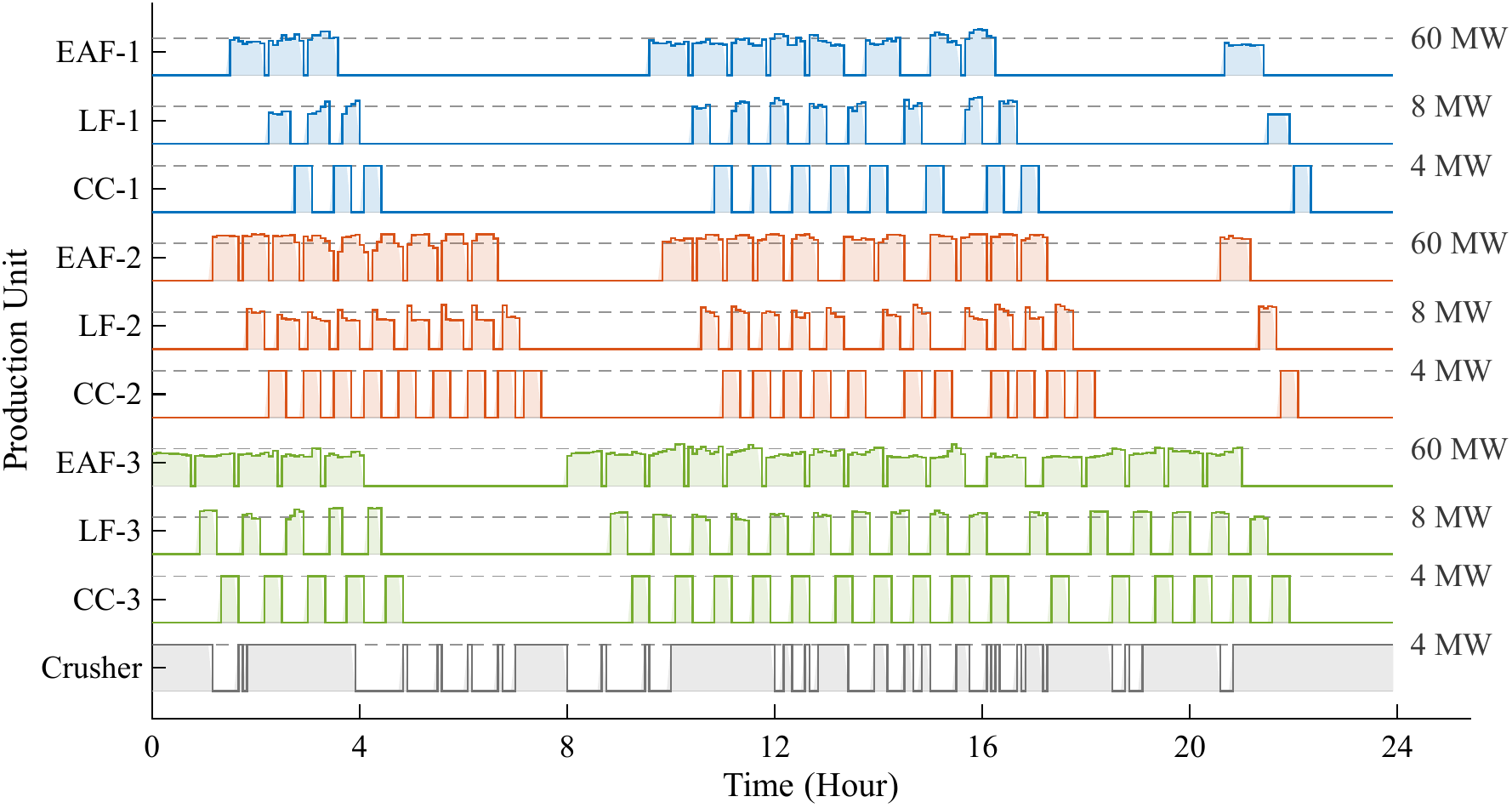}}
    \caption{Unit-level power trajectories generated by the proposed PK-SDRL policy on a representative validation day.}
    \label{Fig_Device_Level_Day_9}
    \vspace{-0.3cm}
\end{figure}

\vspace{-0.5cm}
\subsection{Ablation and Comparative Performance Evaluation}
\label{Ablation and Comparative Performance Evaluation}

\subsubsection{Ablation Study}
\label{Safety-Aware Ablation}

To examine the necessity of the PDG-AP mechanism under tightly constrained process scheduling, a standard DRL baseline is constructed by removing this mechanism while retaining the same state space, action space, reward function, network architecture, and training data as PK-SDRL. The training behavior is shown in Fig.~\ref{Raw_DRL}. Although the daily reward of DRL increases during training, the number of completed heats decreases from nearly 54 heats/day to a much lower level, while hot-metal losses remain significant. This result indicates that, without PDG-AP, the agent cannot effectively learn the inter-stage transfer timing, waiting-time limits, and mandatory process-continuity requirements only from posterior penalties. The apparent reward improvement in the later training stage is mainly attributed to fewer high-risk startups and lower electricity consumption, rather than to a feasible production schedule. In other words, the DRL baseline tends to avoid process failures through conservative under-production.

\begin{figure}[htbp] 
    \vspace{-0.1cm}
    \centerline{\includegraphics[width=1\columnwidth]{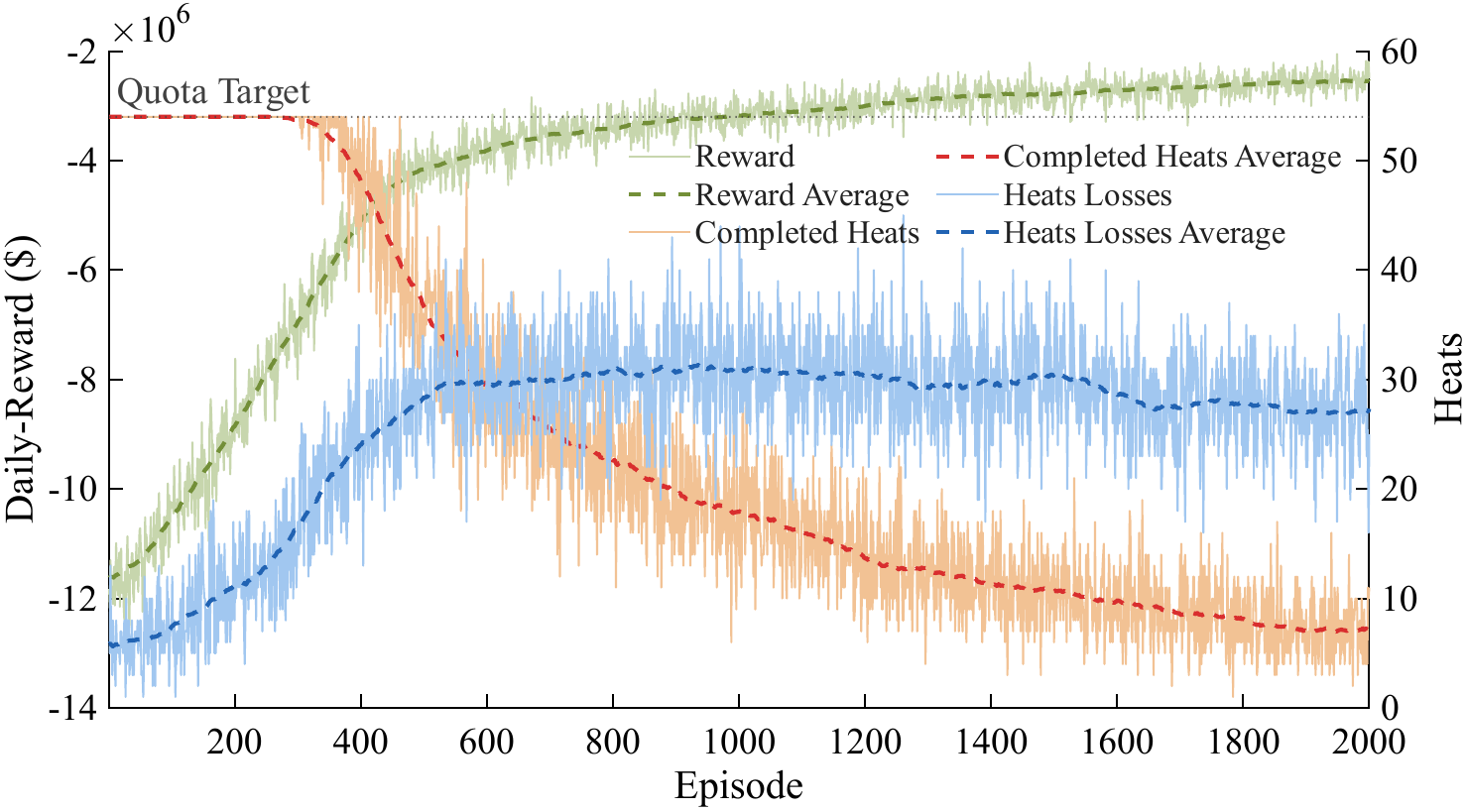}}
    \caption{Training behavior of the standard DRL baseline without PDG-AP.}
    \label{Raw_DRL}
    \vspace{-0.3cm}
\end{figure}

In contrast, as shown in Fig.~\ref{reward}, PK-SDRL achieves 100\% quota fulfillment rate under greedy validation, with the validation reward converging to a stable level. This comparison demonstrates that the PDG-AP excludes clearly infeasible actions and enforces mandatory process actions before execution, thereby restricting policy exploration to the feasible production space.

\subsubsection{Comparative Performance Evaluation}
\label{Comparative Performance Evaluation}

The proposed PK-SDRL method is further compared with a Rule-Based Policy and Rolling MILP. The Rule-Based Policy completes the production task following a fixed and uniform production pace, while Rolling MILP solves a batch-level mixed-integer scheduling problem at each time step with a 3-h prediction horizon and a 5-min rolling interval. Fig.~\ref{fig:method_comparison_hourly} presents the hourly power and electricity-cost profiles on a representative validation day. The Rule-Based Policy follows a nearly fixed production pattern and has limited capability to adjust the load distribution according to electricity prices and renewable generation, leading to high production power during some high-cost periods. Rolling MILP exhibits certain price-responsive behavior, but its short 3-h prediction horizon and local heat-completion revenue still result in temporally concentrated production, which increases the hourly electricity cost in several periods. In contrast, PK-SDRL adjusts the real-time dispatch more flexibly while satisfying the daily production target and process constraints. It allocates high-power operation more effectively to periods with lower grid-import pressure, thereby reducing the overall operating cost. The full-validation-set statistics are reported in Table~\ref{tab:comparison_results}. Compared with short-horizon Rolling MILP, PK-SDRL also avoids the need for frequent online solution of large-scale mixed-integer optimization problems and exhibits superior economic scheduling performance under dynamic electricity prices and renewable-generation fluctuations.

\begin{figure}[htbp] 
    \vspace{-0.3cm}
    \centerline{\includegraphics[width=1\columnwidth]{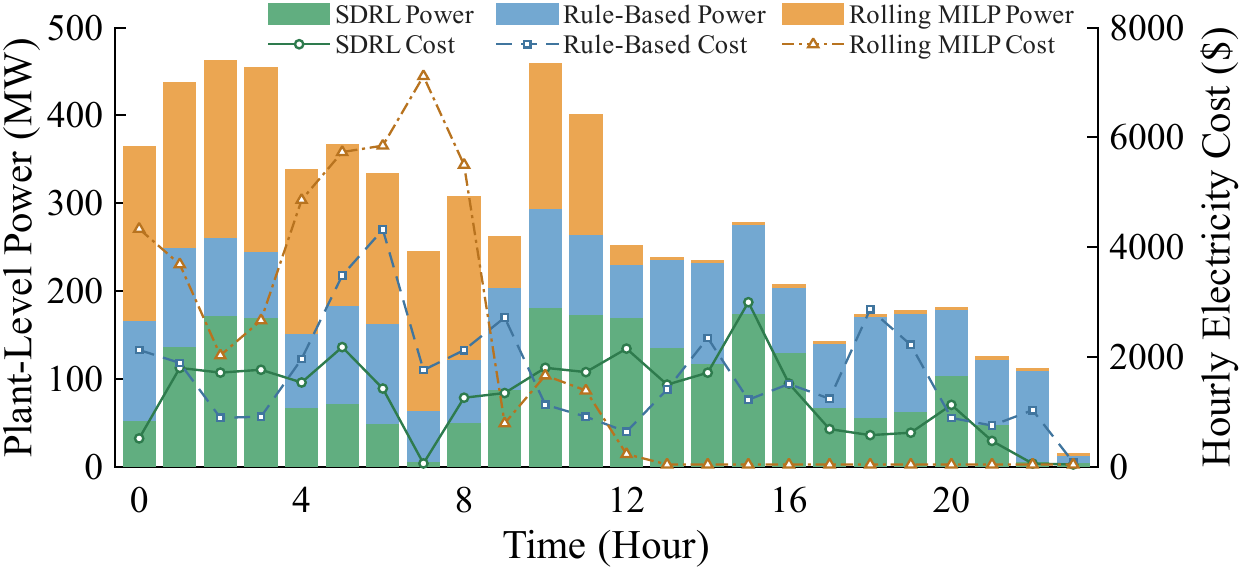}}
    \caption{Hourly power and electricity-cost comparison of different scheduling methods on a representative validation day.}
    \label{fig:method_comparison_hourly}
    \vspace{-0.3cm}
\end{figure}

\vspace{-0.2cm}
\begin{table}[htbp]
\caption{Performance Comparison of Different Scheduling Methods}
\label{tab:comparison_results}
\centering
\footnotesize
\setlength{\tabcolsep}{3pt}
\renewcommand{\arraystretch}{1.08}
\resizebox{\columnwidth}{!}{%
\begin{tabular}{@{}lcccc@{}}
\toprule
Method
& Electricity Cost
& Quota Hit 
& Process-Loss
& Decision Time \\
& /day
& Rate 
& Rate 
& /5-min step \\
\midrule
PK-SDRL
&  29422.5\$
&  100\% 
&  0 
&  0.18 ms  \\
DRL
&  \textbackslash  
&  0
&  72.1\% 
&  0.11 ms  \\
Rule-Based Policy
&  57946.3\$
&  100\% 
&  0
&  \textbackslash   \\
Rolling MILP
&  39721.2\$
&  100\%
&  0
&  9.8 s  \\
\bottomrule
\end{tabular}%
}
\vspace{-0.3cm}
\end{table}

\subsection{Sensitivity Analysis}
\label{sensitivity_analysis}

Two sensitivity studies are conducted to evaluate the effects of forecast-error
standard deviation $\sigma_{\mathrm f}$ and forecast horizon $H^{\mathrm f}=H_s\Delta t$ on
PK-SDRL. As shown in Fig.~\ref{fig:sensitivity_analysis}(a), the quota
fulfillment rate remains 100\% for $\sigma_{\mathrm f}\leq20\%$, while the cost
per heat increases with forecast uncertainty. At $\sigma_{\mathrm f}=50\%$,
severely distorted price and renewable-generation forecasts prevent the policy
from reliably identifying low-cost production windows, causing delayed
operations and reducing the completed output to approximately 50 heats on a
few validation days. The daily quota is imposed through a terminal penalty
rather than a hard constraint or a predefined end-of-day catch-up rule. Such
mechanisms would substantially restrict scheduling flexibility, while a
universally valid catch-up instant cannot be specified under heterogeneous
daily price, renewable-generation, and process-state trajectories.

Fig.~\ref{fig:sensitivity_analysis}(b) shows that increasing $H^{\mathrm f}$
reduces the cost per heat, whereas the corresponding online decision time
increases only marginally and remains well below the 5-min control interval.
The results confirm that higher forecast accuracy and longer look-ahead
information both improve economic performance. Nevertheless, the marginal
cost reduction becomes limited beyond $H^{\mathrm f}=3$~h, while reliable
forecasts at a 5-min resolution over substantially longer horizons are difficult
to obtain in practical operation.

\begin{figure}[htbp] 
    \vspace{-0.0cm}
    \centerline{\includegraphics[width=1\columnwidth]{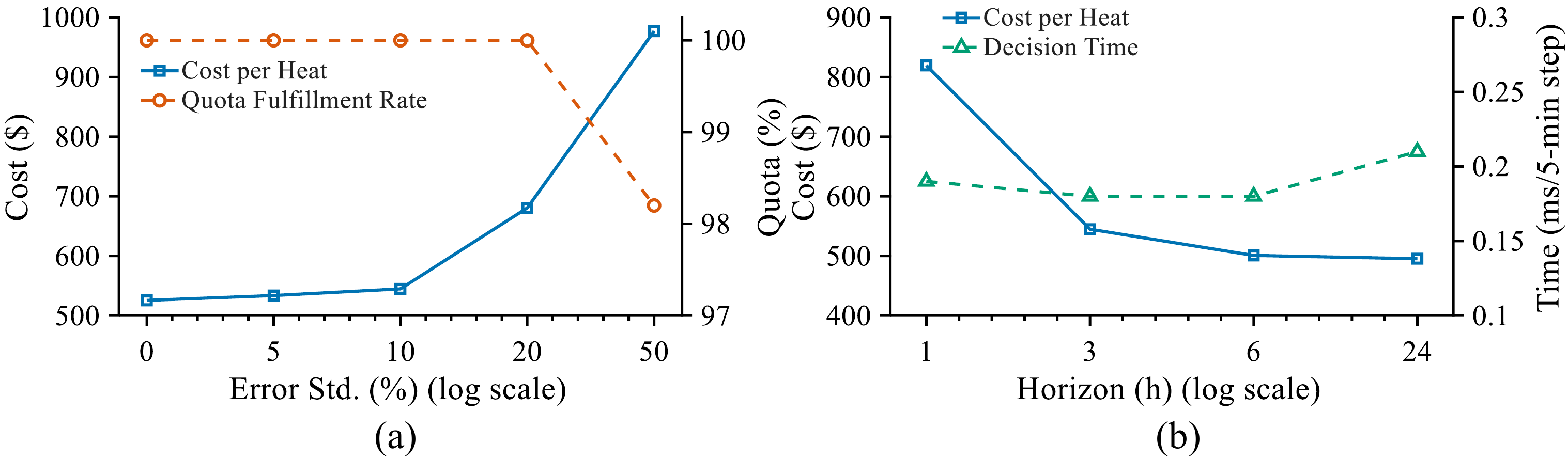}}
    \caption{Sensitivity analysis of PK-SDRL: (a) performance under different forecast-error standard deviations and (b) performance under different forecast-horizon lengths.}
    \label{fig:sensitivity_analysis}
    \vspace{-0.5cm}
\end{figure}

\section{Conclusion}\label{CONCLUSION}

In this paper, a process-knowledge-embedded safe deep reinforcement
learning framework is proposed for the 5-min real-time
dispatch of SPLs under a two-part electricity tariff and local
renewable-generation uncertainties. The proposed PDG-AP mechanism
reallocates excluded-action probabilities according to process distance
and the actor's safe-action preference, while recursive process
feasibility ensures admissible execution and feasible continuation.
The expected process-correction distance is further incorporated into
parameterized-action PPO through a process-correction budget and a
primal--dual update, and the derived bound quantifies the raw policy's
dependence on safety processing. Together with the lossless
active-frontier action space, state-transition shaping reward, and
terminal quota reward, PK-SDRL coordinates economic operation,
process advancement, and daily quota fulfillment. Case studies
on a multiline EAF--LF--CC steel plant demonstrate zero process loss. Compared with the rule-based policy and
rolling MILP, PK-SDRL reduces daily electricity costs by 49.2\% and
25.9\%, respectively, while requiring only 0.18~ms per decision step.
Sensitivity studies further demonstrate its effectiveness under
forecast errors and different forecast horizons.

Future work will extend the proposed framework to more complex
industrial processes and investigate its generalization under broader
operating conditions and uncertainty patterns.

\vspace{-0.3cm}
\appendix

\subsection{Proof of Proposition~\ref{prop:pdgam_recursive_feasibility}}
\label{appendix:pdgam_recursive_feasibility}

\begin{proof}
Consider an arbitrary reachable state $\Xi_{t-1}$ and any
$a_t^x\in\mathcal U_t^{\mathrm{s}}(\Xi_{t-1})$, with
$\Xi_t=F(\Xi_{t-1},a_t^x)$. Define the processing duration and
interstage interval as:
\begin{equation}
d_k^{\ell,s}=
\left(
t_{k,\mathrm{fn}}^{\ell,s}
-
t_{k,\mathrm{st}}^{\ell,s}
+
1
\right)\Delta t,
q_k^{\ell,s,r}=
\left(
t_{k,\mathrm{st}}^{\ell,r}
-
t_{k,\mathrm{fn}}^{\ell,s}
-
1
\right)\Delta t,
\end{equation}

\noindent where
$\underline{\tau}_k^{\ell,s}\le d_k^{\ell,s}
\le\overline{\tau}_k^{\ell,s}$ and
$\underline{\tau}_{\mathrm{tr}}^{s,r}\le q_k^{\ell,s,r}
\le\overline{\tau}_{\mathrm{tr}}^{s,r}$.
Preserve all committed operations and, for each uncommitted heat $k+1$,
select: $q_{k+1}^{\ell,\mathrm{EAF},\mathrm{LF}}
=\overline{\tau}_{\mathrm{tr}}^{\mathrm{EAF},\mathrm{LF}},\;
q_{k+1}^{\ell,\mathrm{LF},\mathrm{CC}}
=\overline{\tau}_{\mathrm{tr}}^{\mathrm{LF},\mathrm{CC}},\;
d_{k+1}^{\ell,\mathrm{LF}}
=\overline{\tau}_{k+1}^{\ell,\mathrm{LF}}.$ For any $\ell\in\mathcal L$ and consecutive heats
$k,k+1\in\mathcal K_\ell$, exclusive EAF occupancy gives:
\begin{equation}
\left(
t_{k+1,\mathrm{fn}}^{\ell,\mathrm{EAF}}
-
t_{k,\mathrm{fn}}^{\ell,\mathrm{EAF}}
\right)\Delta t
\ge
d_{k+1}^{\ell,\mathrm{EAF}}
\ge
\underline{\tau}_{k+1}^{\ell,\mathrm{EAF}}.
\end{equation}

Hence, the LF separation satisfies:
\begin{equation}
\begin{aligned}
&
\left(
t_{k+1,\mathrm{st}}^{\ell,\mathrm{LF}}
-
t_{k,\mathrm{fn}}^{\ell,\mathrm{LF}}
-
1
\right)\Delta t\\
&=
\left(
t_{k+1,\mathrm{fn}}^{\ell,\mathrm{EAF}}
-
t_{k,\mathrm{fn}}^{\ell,\mathrm{EAF}}
\right)\Delta t
+
q_{k+1}^{\ell,\mathrm{EAF},\mathrm{LF}}
-
q_k^{\ell,\mathrm{EAF},\mathrm{LF}}
-
d_k^{\ell,\mathrm{LF}}\\
&\ge
\underline{\tau}_{k+1}^{\ell,\mathrm{EAF}}
-
\overline{\tau}_{k}^{\ell,\mathrm{LF}}
\ge 0.
\end{aligned}
\label{eq:appendix_lf_nonoverlap}
\end{equation}

Similarly, the CC separation satisfies:
\begin{equation}
\begin{aligned}
&
\left(
t_{k+1,\mathrm{st}}^{\ell,\mathrm{CC}}
-
t_{k,\mathrm{fn}}^{\ell,\mathrm{CC}}
-
1
\right)\Delta t\\
&=
\left(
t_{k+1,\mathrm{fn}}^{\ell,\mathrm{EAF}}
-
t_{k,\mathrm{fn}}^{\ell,\mathrm{EAF}}
\right)\Delta t
+
q_{k+1}^{\ell,\mathrm{EAF},\mathrm{LF}}
-
q_k^{\ell,\mathrm{EAF},\mathrm{LF}}\\
&\quad+
d_{k+1}^{\ell,\mathrm{LF}}
-
d_k^{\ell,\mathrm{LF}}
+
q_{k+1}^{\ell,\mathrm{LF},\mathrm{CC}}
-
q_k^{\ell,\mathrm{LF},\mathrm{CC}}
-
d_k^{\ell,\mathrm{CC}}\\
&\ge
\underline{\tau}_{k+1}^{\ell,\mathrm{EAF}}
+
\overline{\tau}_{k+1}^{\ell,\mathrm{LF}}
-
\overline{\tau}_{k}^{\ell,\mathrm{LF}}
-
\overline{\tau}_{k}^{\ell,\mathrm{CC}}
\ge 0.
\end{aligned}
\label{eq:appendix_cc_nonoverlap}
\end{equation}

Therefore, the backup continuation causes no EAF, LF, or CC conflict,
and:
\begin{equation}
a_{t+1}^{x,\mathrm{bk}}
\in
\mathcal U_{t+1}^{\mathrm{s}}(\Xi_t),
\quad
\mathcal U_{t+1}^{\mathrm{s}}(\Xi_t)\neq\varnothing.
\end{equation}

Since this holds for every
$a_t^x\in\mathcal U_t^{\mathrm{s}}(\Xi_{t-1})$ and
$\tilde{\pi}_{\theta}^{x}$ is supported on
$\mathcal U_t^{\mathrm{s}}(\Xi_{t-1})$:
\begin{equation}
\begin{aligned}
&\Pr\!\left(
\left\{
a_t^x\in\mathcal U_t^{\mathrm{s}}(\Xi_{t-1})
\right\}
\cap
\left\{
\mathcal U_{t+1}^{\mathrm{s}}(\Xi_t)\neq\varnothing
\right\}
\,\middle|\,
\mathbf s_t
\right)\\
&=
\sum_{a\in\mathcal U_t^{\mathrm{s}}(\Xi_{t-1})}
\tilde{\pi}_{\theta}^{x}(a\mid\mathbf s_t)
=1.
\end{aligned}
\end{equation}

\noindent which proves Proposition~\ref{prop:pdgam_recursive_feasibility}.
\end{proof}

\bibliographystyle{IEEEtran}
\bibliography{VESS}

\end{document}